\documentclass[11pt,letterpaper]{article}

\usepackage[normalem]{ulem}
\usepackage[margin=1in]{geometry}
\usepackage{lmodern}
\usepackage[T1]{fontenc}
\usepackage[utf8]{inputenc}
\usepackage[tbtags]{amsmath}
\allowdisplaybreaks
\usepackage{bbm}
\usepackage{amssymb,amsthm}
\usepackage[linktocpage=true,breaklinks,colorlinks =false,hidelinks,bookmarks=false,pagebackref=true]{hyperref}

\usepackage{doi}
\usepackage{hyphenat}
\usepackage{appendix}
\usepackage[capitalise,nameinlink]{cleveref}

\usepackage{dsfont}
\usepackage{microtype,bm}
\usepackage{booktabs}
\usepackage[square,sort,comma,numbers]{natbib}
\usepackage{xcolor}
\usepackage{setspace}	
\usepackage{nicefrac}
\usepackage{xspace}

\newtheorem*{theorem*}{Theorem}
\usepackage{thm-restate}

\newtheorem{theorem}{Theorem}[section]
\newtheorem{lemma}[theorem]{Lemma}
\newtheorem{proposition}[theorem]{Proposition}
\newtheorem{corollary}[theorem]{Corollary}

\theoremstyle{definition}
\newtheorem{definition}[theorem]{Definition}
\newtheorem{claim}[theorem]{Claim}
\crefname{claim}{Claim}{Claims}

\newtheorem{example}[theorem]{Example}

\newenvironment{algo}{\vspace{0.07in} \noindent \begin{minipage}	{\textwidth}%\begin{quote}
\hrule\vspace{0.01in}\hrule \vspace{0.05in}}{\vspace{-0.07in}\hrule\vspace{0.01in}\hrule \vspace{0.07in}%\end{quote} 
\end{minipage}}%

\newcommand{\edit}[1]{{\color{black}#1}}

\newcommand{\E}{\mathbb{E}}
\newcommand{\R}{\mathbb{R}}
\newcommand{\Rp}{{\R}_{+}}

\newcommand{\buyers}{{[n]}}
\newcommand{\items}{{[m]}}

\newcommand{\dist}{D}
\newcommand{\demand}{\mathcal{D}}
\newcommand{\alloc}{A}

\newcommand{\sample}{s}

\newcommand{\sampleprime}{s'}

\newcommand{\realvalue}{r}
\newcommand{\bundle}{B}
\newcommand{\allocprofile}{\mathbf{A}}
\newcommand{\allocations}{\mathcal{A}}
\newcommand{\ALG}{\mathsf{ALG}\xspace}
\newcommand{\GRD}{\mathsf{G}\xspace}
\newcommand{\OPT}{\mathsf{OPT}\xspace}
\newcommand{\GRDTWO}{\mathsf{\hat{G}}\xspace}
\newcommand{\bv}{\mathbf{v}} 
\newcommand{\bp}{\mathbf{p}}
\newcommand{\bs}{\mathbf{s}}
\newcommand{\br}{\mathbf{r}}

\newcommand{\bV}{\mathbf{V}}
\newcommand{\bS}{\mathbf{S}}

\newcommand{\bq}{\mathbf{q}}
\newcommand{\bw}{\mathbf{w}}

\newcommand{\EXP}[2][]{
    \ifthenelse{\equal{#1}{}}
    {\mathbb{E}\left[#2\right]}
    {\mathop{\mathbb{E}}_{#1}\left[#2\right]}
}
\newcommand{\PRO}[2][]{
    \ifthenelse{\equal{#1}{}}
    {\mathrm{Pr}\left[#2\right]}
    {\mathop{\mathrm{Pr}}_{#1}\left[#2\right]}
}

\newcommand{\base}{b}

\begin{document}

%% Title
\title{Online Combinatorial Allocations and Auctions with Few Samples}
%% Authors
\author{ Paul D\"utting\footnote{Google Research, Zurich, Switzerland. Email: duetting@google.com} \and Thomas Kesselheim\footnote{Institute of Computer Science and Lamarr Institute for Machine
Learning and Artificial
Intelligence, University of Bonn, Germany. Email:  thomas.kesselheim@uni-bonn.de}  \and Brendan Lucier\footnote{Microsoft Research New England, Cambridge, MA, USA. Email:  brlucier@microsoft.com} \and \and Rebecca Reiffenh\"auser\footnote{Institute for Logic, Language and Computation, University of Amsterdam, The Netherlands. Email:  r.e.m.reiffenhauser@uva.nl} \and Sahil Singla\footnote{School of Computer Science, Georgia Tech, Atlanta, GA, USA. Email: ssingla@gatech.edu. Supported in part by NSF award CCF-2327010.}}

\date{}

\maketitle

%%%%%ABSTRACT%%%%%
\begin{abstract}
In online combinatorial allocations/auctions, $n$ bidders sequentially arrive,  each with a combinatorial valuation (such as submodular/XOS) over subsets of $m$ indivisible items. The aim is to immediately allocate a subset of the remaining items  to maximize the total \emph{welfare}, defined as the sum of bidder valuations. 
A long line of work has studied this problem when the bidder valuations come from known \edit{independent} distributions.  In particular, for submodular/XOS valuations, we know $2$-competitive algorithms/mechanisms that set a fixed price for each item and the arriving bidders take their favorite subset of the remaining items given these prices.  However, these algorithms traditionally presume the availability of the underlying distributions as part of the input to the algorithm. 
Contrary to this assumption, practical scenarios often require the learning of distributions, a task complicated by limited sample availability. This paper investigates the feasibility of achieving $O(1)$-competitive algorithms under the realistic constraint of having access to only a limited number of samples from the underlying bidder distributions.

Our first main contribution shows that a mere single sample from each bidder distribution is sufficient to yield an $O(1)$-competitive algorithm for submodular/XOS valuations.  This result leverages a novel extension of the secretary-style analysis, employing the sample to have the algorithm compete against itself.  Although online,  this first approach does not provide an online truthful mechanism.  Our second main contribution shows that a polynomial number of samples suffices to yield a $(2+\epsilon)$-competitive online truthful mechanism for  submodular/XOS valuations and any constant $\epsilon>0$.   This result is based on a generalization of the median-based algorithm for the single-item prophet inequality problem to combinatorial settings with multiple items.  
\end{abstract}

\newpage

%\setcounter{tocdepth}{1}
% {\small
% %\begin{spacing}{0}
%    \tableofcontents
% %\end{spacing}
% }

%%%%%INTRODUCTION%%%%%
\section{Introduction}

Online combinatorial allocation is a fundamental problem in stochastic optimization, capturing a wide variety of resource allocation problems.  In such problems, requests arrive online for a common pool of indivisible  items/resources. 
Each request $i$ can be assigned some  combination of items/resources. The satisfaction of the request is  modeled by a valuation function $v_i$ (e.g. unit-demand, submodular, or XOS), drawn \edit{independently} from some  probability distribution $D_i$, quantifying the value of different  combinations. 
A decision-maker must immediately and irrevocably allocate to each request  a subset of the available items/resources, while maximizing  the total sum of valuations. The challenge is that allocation decisions must be made online without perfect foreknowledge of future requests.  Such online allocation problems capture wide-ranging applications, from cloud computation scheduling to e-commerce, and have attracted substantial attention from the theory community \cite[e.g.,][]{FeldmanGL15,DuettingFKL17,DuttingKL20,CorreaC23}.

Importantly, the performance achievable by the decision-maker depends crucially on how much information they have about the request distributions $D_i$. If nothing is known (i.e., the distributions are adversarial) then no bounded approximation to the offline optimal solution is possible even for allocating a single item. {On the other end of the spectrum, when valuations are drawn independently from \emph{known} distributions (the so-called prophet setting), sweeping positive results are possible. For example, in a celebrated result,  Feldman, Gravin, and Lucier \cite{FeldmanGL15} have shown that when the valuations $v_i$ are submodular (or even XOS), it is possible to achieve a (tight) factor $2$-approximation against the expected offline optimum.}

In the prophet setting, many known solutions for online combinatorial allocation have the additional desirable property of being \emph{truthful}, meaning that a self-interested requester (or bidder) can obtain no benefit by strategically manipulating their reported valuation $v_i$.  This is a non-trivial property even when the distributions are fully known in advance, since classical solutions from economics such as the VCG auctions \cite{Vickrey61,Clarke71,Groves73} do \emph{not} work for online bidders. The result of 
Feldman, Gravin, and Lucier \cite{FeldmanGL15}{, for example, achieves truthfulness by taking the form of} a ``posted-price mechanism.'' {This type of mechanism uses the advance knowledge of the valuation distributions to compute fixed item prices, and achieves truthfulness by letting each arriving bidder buy a utility-maximizing bundle of items.}

{
The two extremes discussed so far, no information about the distributions or full information, are respectively too pessimistic and overly optimistic. Indeed, in real-world scenarios, one would expect to have some (but not perfect) distributional knowledge through observational data.  
Motivated by this, Azar, Kleinberg, and Weinberg \cite{AzarKW14} suggest taking a sample complexity perspective and seek guarantees that apply when the underlying distributions are known only through a limited number of samples. 
This has led to a sequence of efforts~\cite{RubinsteinWW20,KaplanNR22,CDFFLLP-SODA22} expanding sample-based prophet inequalities for allocation problems, but fairly little was known about general submodular (and, more generally, XOS) valuations.}
{To address this gap, we ask:}

\begin{quote}
\emph{{For submodular (or even XOS) valuations,} how many samples from the distributions suffice to obtain constant-competitive online allocations?  How many samples are needed if we additionally require the algorithm to be truthful?}
\end{quote}

Given existing work that solves the online allocation problem for known distributions, a natural approach involves first learning the distribution, followed by algorithm creation tailored to the learned distribution. Yet, given the inherent imperfection in distribution learning, this strategy often culminates in suboptimal performance.  Indeed, our investigation indicates that prevailing posted-price mechanisms 
are intrinsically reliant on precise knowledge of distributions. This dependency, in essence, stems from the need to learn the ``mean'' of certain random variables, a feat unattainable with a limited sample set unless we adopt stringent assumptions, such as bounded support for valuations.

An alternative approach, which is rooted in the realm of Data-Driven Algorithm Design \cite{Balcan-Chp20},  proposes bypassing distribution learning and instead directly learning algorithms from samples.  
This approach, especially when applied to online settings with combinatorial assignments, introduces significant complexities even when we set aside truthfulness and concentrate solely on constructing an online allocation algorithm to optimize welfare.

In response to the above questions,  our contributions are twofold:
\begin{enumerate}
\item 
\emph{Online Combinatorial Allocation:} In the non-strategic setting with request valuations that are submodular (or even XOS), we establish that even a single sample from each distribution suffices to obtain an $O(1)$-competitive online allocation algorithm.  {This presents as significant advancement, given that prior single-sample results} were limited to specific submodular valuations like unit-demand or budget-additive \cite{KaplanNR22,CDFFLLP-SODA22}.

\item \emph{Posted-Price Mechanisms}: We prove that, given bidder valuations that are submodular (or more generally XOS),  a polynomial number of samples  can effectively inform the choice of fixed item prices that yield $(2+\epsilon)$-competitive posted-price mechanisms for any constant $\epsilon>0$.  
Given the intrinsic unlearnability of a random variable's ``mean'', this achievement is anchored in pioneering a ``median'' pricing strategy for combinatorial settings---potentially a topic of independent interest. {Prior to this work, it was only known how to obtain truthful, constant-factor approximations with poly-many samples under stringent assumptions on the  distributions, such as boundedness \cite{FeldmanGL15}.}
\end{enumerate}

\subsection{Median Pricing and Secretary-Style Analysis for Single Item}

To build intuition { for our results},  we examine the scenario of a single-item allocation/auction, analogous to the single-choice prophet inequality problem \cite{KrengelS77,KrengelS78}.  Here, we encounter $n$ bidders arriving sequentially, each with a valuation $v_i$ for the item drawn from an independent distribution.   We need to design an online algorithm to maximize the value of the bidder who receives the item.   A well-known strategy achieves a $2$-competitive ratio by assigning the item to the first bidder whose valuation  exceeds half the expected maximum value \cite{KleinbergW12}. The posted-price mechanisms introduced by Feldman, Gravin, and Lucier \cite{FeldmanGL15} extend this mean-based philosophy to the multi-item, combinatorial auction framework.

However, the landscape changes when we possess only samples of the bidders' valuation distributions. Learning the mean of such distributions is infeasible with a sample-based approach, particularly for distributions with small, yet pivotal probabilities. {For instance,  consider an instance with $n=2$ bidders where $v_1 = 1$ w.p.  $1$  and $v_2 = 1/\epsilon$ w.p.  $\delta$ and is $0$ otherwise,  where $\epsilon,\delta$ are   small  unknown constants.    We need at least $\approx 1/\delta$ samples to see a non-zero $v_2$;  however,  $\epsilon,\delta$ could be arbitrarily small,  so we cannot (approximately) learn $\E[\max \{v_1,v_2\}]$ for this distribution. }

Nonetheless, even with limited samples, achieving \(O(1)\)-competitive algorithms is plausible through two principal methods, both circumventing the need to ascertain the distribution means and instead focusing on constant probability selection of bids.

\medskip
\noindent \textbf{Secretary-Style Pricing.}
One methodology is to select a threshold based on the maximum value among a single sample from each bidder's distribution \cite[e.g.,][]{AzarKW14,RubinsteinWW20}. 
The reason this algorithm is $O(1)$-competitive is due to a Secretary-Style argument: 
  among all samples  and real values,  with constant probability,  the maximum is among the real values and the second-highest is among the samples.  Although it has a simple analysis for single-item,  the difficulty of extending this to combinatorial valuations is unclear since there is no natural notion of ``maximum'' and ``second-highest'' due to item dependencies.  Our single-sample $O(1)$-competitive online combinatorial allocation algorithms are based on generalizing this secretary-style approach.

\medskip
\noindent \textbf{Median Pricing.}
Alternatively,  we could  pivot from mean to median valuation estimations.
A classic result of Samuel-Cahn~\cite{SamuelCahn84} shows that selecting the first bidder with value above the median $\tau$  is $2$-competitive,  i.e.,  $\tau$ is the solution to $\mathbb{P}\big[\max v_i \geq \tau \big] = 1/2$.  
With a finite sample set, a threshold \(\hat{\tau}\) can be inferred, which aligns the maximum valuation probability within an acceptable \(\epsilon\) range of the ideal half, thus yielding a \(2+O(\epsilon)\)  competitive algorithm. However, the  challenge in extending median-based approaches to combinatorial auctions is that it's not obvious how to define median for combinatorial valuations due to item dependencies.  
Our polynomial-sample $O(1)$-competitive posted-price mechanism for online combinatorial auctions is  based on generalizing this median pricing approach.

%------------------------------------------------

%------------------------------------------------
\subsection{Online Combinatorial Allocation}

Our first result extends existing single-sample methodologies for the online allocation problem (such as the one for a single item discussed above) all the way to submodular (and even XOS) valuations.  \edit{Note that a single sample from the product distribution over bidder valuations corresponds to drawing one sample from each of the bidder's (independent) valuation distributions.}

\begin{theorem} \label{thm:mainSingleSample}
\edit{Given a single sample from any product distribution over submodular (or even XOS) bidder valuations, 
there exists an $O(1)$-competitive online algorithm for the combinatorial allocation problem.} 
\end{theorem}

We actually prove an even stronger result, showing that the same guarantee can be obtained in the Game of Googol model, introduced in \cite{CorreaCES20}, with $k = 1$ sample per buyer (\Cref{thm:main-single-sample}). 
In this model, an adversary inscribes \(k+1\) valuation functions on \(k+1\) facets of a die for each buyer. The die is then tossed, revealing \(k\) valuation functions to the online algorithm, while the \((k+1)\)-st—the true valuation—remains hidden.  The algorithm must then make online allocations that are competitive with the best possible allocation in expectation over the die rolls. 
Any guarantee made in this model naturally extends to the independent samples model. Moreover, via a reduction from \cite{CDFFLLP-SODA22}, this framework implies the first \(O(1)\)-competitive ``order-oblivious'' secretary algorithm for XOS combinatorial auctions: here, after receiving a constant fraction of the online elements, chosen uniformly at random, as a sample, the leftover elements arrive in adversarial order.

\begin{corollary}\label{cor:xos-oos} There exists an $O(1)$-competitive  order-oblivious secretary algorithm  for the combinatorial allocation problem with submodular (and even XOS) buyers.
\end{corollary}

{
Before describing our techniques,  we discuss why simple approaches don't work. We outline the hardness examples here and discuss them in  detail in \Cref{sec:simpleAlgos}.

\medskip
\noindent\textbf{Limitations of Prior Approaches.}
Probably the most natural idea, inspired from single-item analysis, is to   set the largest observed  value in sample $S$ as a threshold, i.e.,    item $j$'s  price  equals $\max_{i\in S}v_i(j)$. This fails when a single unit-demand buyer has a large value for every item but  has only a small contribution to the optimum welfare.  
A second idea would be to compute the offline optimum allocation for the sample, then price each item according to its ``supporting price'' under that allocation.\footnote{For unit-demand buyers, the supporting price of an item $j$ allocated to agent $i$ is $v_i(j)$.  See Section~\ref{sec:prelims} for the general definition.}
This  doesn't work (even for deterministic valuations) since a unit-demand buyer that contributes most of the welfare might now purchase a suboptimal item at such a high price. 

Motivated by balanced pricing \cite{FeldmanGL15,DuettingFKL17}, a third idea is to instead compute an offline optimal allocation on the single sample, and then set each item's price to be half of its supporting price.  This also does not work.  Suppose $n=m$ and consider a distribution $D$ over profiles of unit-demand valuations, such that player $n$ has value $1$ for one of the items (uniformly at random), and value $\epsilon \ll 1/n$ for the remaining items.  Each other player has value $\epsilon$ for every item. Then this pricing approach will result in price $1/2$ for one of the items (uniformly at random), say item $j$, and price $\epsilon/2$ for all other items.  On the real valuation, the optimal allocation has value $1 + O(n\epsilon)$  but unless $j = i$, player $n$ will receive no items, so the total value obtained will be $O(n\epsilon)$.

\medskip

Given the limitations of prior approaches, we take a totally different approach by first considering a simplified setting where items can be reallocated to the current bidder.
}

\medskip
\noindent\textbf{Greedy with Reallocations.} 
Consider a simplified setting in which we are permitted to reassign items previously allocated to the \(i\)-th bidder. This relaxation substantially simplifies the problem, allowing a straightforward constant approximation \emph{bidder-wise greedy} algorithm, known since at least \cite{DobzinskiNS05}.  Here, each $i$-th bidder may claim any item but must compensate for the value lost by the item's prior owner. In essence, the $i$-th bidder encounters item prices $\bp^{(i)}$, with each item's price reflecting its current contribution to the welfare as determined by its supporting XOS valuation, and may acquire any item (including previously allocated) by remunerating this price. This approach is effective because any highly-valued item can always be repurchased, ensuring that either the current or the previous bidder holds significant valuation for it.

\medskip
\noindent\textbf{Two-Samples Algorithm.} 
 Implementing the buyer-wise greedy strategy online without the luxury of reallocating items poses the principal challenge. Let's approach this by presuming the availability of two samples from each distribution. 
We start by using the first sample to establish item prices, employing the buyer-wise greedy algorithm on this sample to set the prices $\bp^{(i)}$ that the $i$-th real bidder  faces.  The underlying thought is that the `marginal' value ascribed by both the $i$-th sample and the real $i$-th bidder should be analogous, since they face the same prices and are drawn from the same distribution.   However, this alone does not preclude the possibility of item reallocation. Ideally, we wish to allocate the $i$-th bidder a subset of these items that are unlikely to be reassigned subsequently.

\medskip
\noindent\textbf{Competing with Yourself: Secretary-Style Argument.} 
Our foremost insight is the conceptualization of the algorithm as \emph{competing against itself}. While still utilizing the first sample to set item prices \(\bp^{(i)}\) for the \(i\)-th bidder, we observe the price dynamics on the \emph{second sample} with reallocations.  Specifically, we determine which bidder secures each item \(j\) post-reallocations in the second sample, and use the supporting value of this bidder  to set a `base price' \(b_j\) for every item \(j\). As the actual \(i\)-th bidder arrives, they aim to select their optimal bundle \(\bundle_i\) based on the initial sample prices \(\bp^{(i)}\) (assuming full item availability), yet we only assign them a subset of these items that surpass the base prices \(b_j\) and remain unclaimed.  
The intuition is that since the algorithm on the second-sample and the real-sample are facing the same prices,  a secretary-style argument  should imply that with constant probability only the max-bidder for an item $j$ will be allocated the item.

\medskip
\noindent\textbf{Max vs. Sum.} 
Although utilizing the second-sample to set base prices averts reallocation, it poses a risk of diminishing welfare relative to the buyer-wise greedy model.    More precisely,  our \cref{lem:ALG_LB} only implies that the expected contribution of any item $j$ to the welfare of the algorithm  is at least a constant fraction of the \emph{maximum} of the max $j$-bid in the real valuations and in the second sample..  Conversely,  we only know that the \emph{sum} of the $j$-bids in the real valuations and the second sample  is at least a constant fraction of  the buyer-wise greedy algorithm.  
Our key technical contribution, encapsulated in \cref{lem:MaxVsSum}, demonstrates that rounding prices to powers-of-two allows the max and the sum of bids to be comparable. The lemma's proof employs a carefully constructed Martingale argument to circumvent dependencies that might emerge between the samples and actual valuations, such as those arising from base pricing.

Finally,  we extend this proof to single sample  by observing that any $k$-sample result can be turned into a single sample result by losing only an $O(k)$ term (\Cref{thm:reduction}).

%------------------------------------------------
\subsection{Posted-Price Mechanisms}

Our second main result addresses the design of truthful posted-price mechanisms using only sample access to bidder distributions. 
 As mentioned earlier,  posted-price mechanisms are highly desirable due to their practical appeal.  
 Existing posted-price mechanism designs depend on mean-based pricing, which is untenable to implement with unknown and potentially unbounded distributions.
{ Our single-sample approach from the previous section is inherently non-truthful since it provisionally assigns a utility-maximizing set of items to a bidder assuming all items are available, but ultimately only assigns them a subset of the items that are still available.  }

 We bypass these issues by establishing a median-based pricing system for combinatorial settings. 
 In the single-item case, the median price $\tau$ satisfies the condition \(\mathbb{P}[\max v_i \geq \tau] = 1/2\). In combinatorial settings, we will likewise seek prices such that each item sells with probability $1/2$.  We then use median-based pricing to design posted-price mechanisms from samples.

\begin{theorem} \label{thm:mainPolySamples}
Given polynomially many samples from 
any product distribution 
\edit{over submodular (or even XOS) bidder valuations and a}  constant $\epsilon>0$, there exists a $(2+\epsilon)$-competitive truthful posted-price mechanism  for online combinatorial auctions. 
\end{theorem}

We note that the approximation factor in this theorem is nearly tight. This is because even when considering the sale of a single item with known product distributions, online combinatorial auctions capture  single item prophet inequality, where the tight competitive ratio is $2$.

An important subtlety in posted-price mechanisms is the handling of agent indifferences. Even with a single item, if the value distributions have atoms, the exact probability of sale may vary based on an agent's decision-making when the price matches her valuation. To navigate the complexities tied to these tie-breaking scenarios, we initially consider \emph{generic} valuation distributions. This term refers to a generalization of the atomless condition, ensuring that for any specified price vector $\bp$, occurrences of indifference are extremely unlikely. This assumption is practical and realistic, as all distributions become generic with the introduction of even minimal noise. Nonetheless, we later extend our results to include general distributions. For the sake of clarity, we will temporarily set aside the issue of tie-breaking.

\medskip
\noindent\textbf{Combinatorial Median Prices. }
The median price in single-item sales ensures that the item is sold exactly half the time, and unsold for the remainder. We extend this concept to define a median price vector $\mathbf{p}$ for combinatorial settings such that using $\mathbf{p}$ in a posted-price mechanism ensures every item $j$ is sold with a probability of exactly one-half.  The existence of such a median price vector is not immediately apparent, but in Section~\ref{sec:medianPrices} we apply the Kakutani fixed-point theorem to affirm that median price vectors always exist.
%\snote{for generic distributions}.
The reason median prices are interesting is because they possess several amazing properties that we describe next.

\medskip
\noindent\textbf{Posted-Pricing with Median Prices. } The first nice property of median prices that we discover is that they  yield a posted-price mechanism that is $O(1)$-competitive for auctions with submodular (and even XOS) buyers.  
This result greatly generalizes  the classical Optimal Stopping Theory result of Samuel-Cahn \cite{SamuelCahn84}   for the case of a single-item.  The high-level intuition why median prices help   is that they ensure any item $j$ will be available (not sold before) for its  optimal bidder $i$  with at least $1/2$ probability.  Thus,  if item $j$'s median price is low then bidder $i$ can easily purchase it to imply good welfare.  Else,  when $j$'s median price is very high,  we can again use the median-price property that $j$ gets sold with  at least $1/2$ a probability,  so the revenue (and hence welfare)  from $j$ should be large.   
The adaptation of this intuitive understanding to a formal proof requires careful consideration; e.g.,  since the bidders are  stochastic,  the optimal bidder $i$ for any fixed item $j$ is a random bidder and we cannot directly apply the above proof argument.

\medskip
\noindent\textbf{Learning Median Prices. } 
Another beneficial characteristic of median prices is their learnability in combinatorial auction settings with submodular/XOS valuations. We establish that a polynomial number of samples are sufficient to simultaneously approximate the sale probability for each item associated with every price vector \(\bp\) within a small margin of $\pm \epsilon$. This proof employs uniform convergence principles and involves bounding the VC dimension for learning the performance of all price vectors. Our argument for bounding the VC dimension builds on the techniques of \cite{BalcanSV18,BalcanDDKSV21}. Consequently, we demonstrate that the median prices based on empirical data closely estimate those for the underlying unknown distributions.

Finally, we shift our focus to the computational aspects of implementing posted-price mechanisms with combinatorial median prices. To illustrate the usefulness of our approach, we present efficient fully polytime algorithms for unit-demand bidders.

%--------------------------------------------------------------------------------------------------
\subsection{Further Related Work}

\noindent\textbf{Prophet Inequality and Combinatorial Extensions.} Our work and questions are closely related to the literature on prophet inequalities. 
{  The classic prophet inequality problem corresponds to online combinatorial auction with a single item (i.e., $m=1$).}
Prophet inequalities were greatly studied in the Optimal Stopping community in the late 70s early 80s, and the basic prophet inequality problem was solved by  Krengel and Sucheston (who also credit Garling) \cite{KrengelS77,KrengelS78} by giving a tight $2$-approximation. Later,  Samuel-Cahn \cite{SamuelCahn84} showed that this optimal guarantee can be obtained with simple price-based algorithms.

Pioneering work in Computer Science by \cite{HajiaghayiKS07} and  \cite{ChawlaHMS10} demonstrated the usefulness of the prophet inequality paradigm for algorithmic mechanism design.
The connection motivated extensions of the prophet inequality paradigm to combinatorial settings such as matroids \cite{KleinbergW12,FeldmanSZ16} or combinatorial auctions \cite{FeldmanGL15,DuettingFKL17}. These extensions also established the two main frameworks for proving prophet inequalities, namely \emph{online contention resolution} \cite{FeldmanSZ16} and \emph{balanced pricing} \cite{KleinbergW12,FeldmanGL15,DuettingFKL17}. 

Most relevant to our work is the aforementioned work of \cite{FeldmanGL15}, who give a tight $2$-approximation for combinatorial auctions with XOS buyers. This algorithm is based on static item prices, and is therefore truthful. The same approach yields a $O(\log m)$ approximation for subadditive buyers.  Two recent breakthroughs for combinatorial auctions with subadditive buyers, first improved this bound to $O(\log \log m)$ \cite{DuttingKL20}, and then to $O(1)$ \cite{CorreaC23}. 

{ Very recently, \cite{Banihashem-HKKO24} introduced a general approach for turning any prophet inequality algorithm into a pricing-based truthful mechanism. 
For combinatorial auctions, given a non-truthful algorithm as input, this implies a bundle-pricing algorithm (instead of item pricing). This reduction requires a polynomial number of samples  in the size of  outcome space, which is exponential in $m$ for combinatorial auctions.
}

\medskip
\noindent\textbf{Prophet Inequalities with Few Samples.} The study of prophet inequalities from samples was pioneered in \cite{AzarKW14}. Their main result is an approximation-factor preserving reduction from single-sample prophet inequalities to  \emph{order-oblivious} secretary problems. 
See the book chapter \cite{GS-Book20} for further discussion on order-oblivious secretary algorithms. Using this, \cite{AzarKW14}  obtain constant-sample $O(1)$-competitive prophet inequalities for several (restricted) matroid settings and for constant-degree bipartite matching with edge arrivals. They also obtain a $4$-competitive single-sample algorithm for the single choice problem. 
An $e$-competitive secretary algorithm for the online combinatorial allocation problem with XOS buyers was given in \cite{KRTV-ESA13}. However, prior to \Cref{cor:xos-oos}, no order-oblivious secretary algorithm was known for this problem.

{ Starting with \cite{CorreaDFS19} and \cite {RubinsteinWW20} a significant amount of work has examined the single-choice/single-item prophet inequality problem with samples. Key findings of this line of work include that the factor-$2$ single-choice prophet inequality can be attained with only a single sample per buyer (even in the stronger Game of Googol model) \cite{RubinsteinWW20}, while a constant number of samples suffice to get arbitrarily close to the worst-case approximation guarantees for the iid version, the random order, and free order variant of the problem \cite{RubinsteinWW20,GuoEtAl21,CristiZilliotto24}.

In independent work,  \cite{KaplanNR20} introduced the \emph{competitive analysis with a sample} framework. In this framework, an adversary first writes down $m = h + n$ numbers. Afterwards, a randomly chosen subset of size $h$ of these numbers is revealed to the algorithm, who then either observes the remaining $n$ numbers in adversarial order or in random order. The goal is to be competitive with the expected maximum value in this second set of numbers. For the case where $h/n \geq 1$, they give a tight $2$-approximation for the adversarial order version of the problem. For the case where the actual values arrive in random order they give a $(e+1)/e$ approximation.}

Closely related to our work, \cite{CDFFLLP-SODA22} study the single-sample problem for combinatorial settings. They present a general framework for deriving single sample results from greedy algorithms and obtain $O(1)$-competitive prophet inequalities for matchings with vertex or edge arrivals. 
Beyond this, they obtain improved single-sample constant-factor guarantees for a variety of (restricted) matroids and for combinatorial auctions with budget additive bidders. Budget-additive valuations are a strict subclass of submodular valuations, which are themselves a strict subclass of XOS valuations.
In another closely related independent paper, \cite{KaplanNR22} (also see the arXiv version for improved bounds) show $O(1)$-competitive algorithms across a range of the aforementioned settings.

For the general case of online combinatorial auctions with XOS buyers, an $O(1)$-approximation up to an additive error term of $\epsilon$ can be obtained with $\mathsf{poly}(n,m,1/\epsilon)$ samples \cite{FeldmanGL15,DuettingFKL17}. This multiplicative/additive approximation guarantee  turns into a purely multiplicative guarantee only under stringent boundedness assumptions.  {In addition, a \emph{non-truthful} purely multiplicative constant-factor approximation for XOS bidders with $O(n)$ samples per bidder was known (e.g., \cite{CorreaC23}).} The bounds of \cite{AS-FOCS19,AKS-SODA21} for XOS bidders arriving in a random order can be interpreted as giving a truthful $poly(\log\!\log m)$-approximation with a single sample. However, it's unclear how to extend them to an $O(1)$-approximation, even with more more samples.
Finally, \cite{DuttingK19} study prophet inequalities and posted-pricing when we are given distributions that are close to the actual distributions (in some metric).

%%%%%PRELIMINARIES%%%%%
\section{Preliminaries}
\label{sec:prelims}

\noindent\textbf{Online Combinatorial Auctions.} 
We have a set of buyers $\buyers$, a set of items $\items$, and valuation functions $v_i: 2^{\items} \rightarrow \Rp$ that assign each buyer $i \in \buyers$ a value for each subset of items. We assume that the valuation functions are normalized, so that $v_i(\emptyset) = 0$ for all $i \in \buyers$, and monotone increasing, i.e., for every $i \in \buyers$, $S, T \subseteq \items$ with $S \subseteq T$ it holds that $v_i(S) \leq v_i(T)$.
We write $\bv = (v_1, v_2, \ldots, v_n)$ for a valuation profile. 
We further assume that each valuation function $v_i$ is independently drawn from an according, unknown distribution $\dist_i$, and call the tuple of all $\dist_i$ distribution $\dist$.

An assignment of the items to the buyers is a collection of sets $\alloc_i \subseteq \items$ such that $\alloc_i \cap \alloc_j = \emptyset$ for all $i,j \in \buyers$ with $i \neq j.$ 
We write $\allocprofile$ for the tuple of $\alloc_i$ with $i \in \buyers$, and $\allocations$ for all possible allocations. 
The (social) welfare of an assignment is $\sum_{i \in \buyers} v_i(\alloc_i)$.

Our goal is to find a welfare-maximizing assignment of items to buyers, when items are available offline but buyers arrive online in some (previously known, but adversarially chosen) order. 
We assume for convenience that buyer $i$ will arrive at time $i$, for $i\in \buyers$.

\medskip\noindent\textbf{Competitive Ratio.} 
Our focus in this work is on online algorithms $\ALG$ for this problem that have sample access to $\dist$. For this we will assume that $\ALG$ receives $k$ sample valuation functions for each buyer $i \in \buyers$, sampled independently from $\dist_i$. We will typically refer to these as $s_i^1, s_i^2, \ldots, s_i^k$, and use $v_i$ for the actual valuation function. (We will also sometimes write $S_i^1, S_i^2, \ldots, S_i^k$ and $V_i$ to emphasize that these are random variables.)

Let's denote by $\ALG_i(\bv; \bs^1, \ldots, \bs^k)$ the allocation to buyer $i \in \buyers$ on input $\bv = (v_1, v_2, \ldots, v_n)$ and samples $\bs^\ell = (v_1^\ell, \ldots, v_n^\ell)$ for $\ell \in [k]$. We say that $\ALG$ is $\alpha$-competitive if, for all distributions $\dist$, it holds that
\begin{align*}
\EXP[\bv \sim \dist, (\bs^1, \ldots, \bs^k) \sim \dist^k]{\sum_i v_i(\ALG_i(\bv; \bs^1, \ldots, \bs^k))} \geq \alpha \cdot \EXP[\bv \sim \dist]{\max_{\allocprofile \in \allocations} \sum_{i \in \buyers} v_i(\alloc_i)}.
%\label{eq:competitive-ratio}
\end{align*}
The existence of such an online algorithm $\ALG$ establishes an \emph{$\alpha$-competitive $k$-sample prophet inequality} (or $\alpha$-competitive $k$SPI).

\medskip\noindent\textbf{Game of Googol.} We will also employ a somewhat stronger model,  
the Game of Googol with a $(k+1)$-faceted dice. 
In this model an adversary, for each buyer $i \in \buyers$, writes down $k+1$ valuation functions $v_i^1, v_i^2, \ldots, v_i^{k+1}$, and then these valuation functions are assigned to $S_i^1, S_i^2, \ldots, S_i^k$ and $V_i$ uniformly at random. The $S_i^j$ are the samples, and $V_i$ is the actual valuation. The goal is then to show that for any possible choice of the adversary,
\begin{align*}
\EXP{\sum_i v_i(\ALG_i(\bV; \bS^1, \ldots, \bS^k)} \geq \alpha \cdot \EXP{\max_{\allocprofile \in \allocations} \sum_{i \in \buyers} v_i(\alloc_i)}.
%\label{eq:competitive-ratio}
\end{align*}
All other aspects of the problem, including the fixed (worst-case) arrival order, remain untouched.
We refer to such a result as an $\alpha$-competitive pointwise $k$-sample prophet inequality (or an $\alpha$-competitive P-$k$SPI). 
Observe that the existence of an $\alpha$-competitive P-$k$SPI implies an $\alpha$-competitive $k$SPI (simply because instead of choosing the valuations adversarially, we could also have drawn them from distributions $\dist$).

\medskip\noindent\textbf{Classes of Valuations.} All our valuation functions $v: 2^{[m]} \rightarrow \Rp$  are monotone, i.e., for any item bundles $A \subseteq B \subseteq [m]$, we have $v(A)\leq v(B)$. 
We mostly focus on XOS valuation functions \cite{Feige09}, but also briefly touch upon the subclasses of submodular and unit-demand valuations: 
\begin{itemize}
\item \textbf{XOS valuations:} Valuation function $v: 2^\items \rightarrow \Rp$ is XOS if there exist a collection of additive functions $\{a_{\ell}: 2^\items \rightarrow \Rp\}_{\ell=1,\ldots,k}$ such that for each bundle of items $B \subseteq \items$ it holds that $v(B) = \max_{\ell = 1, \ldots, k} a_\ell(B) = \max_{\ell = 1, \ldots, k} \sum_{j \in B} a_{\ell}(j)$.\\
For an XOS valuation function $v$ we refer to the additive function that defines $v$ on the set of items $B \subseteq \items$ as the additive supporting function of $v$ on $B$. Let's denote this function by $a_B$. Note that this function satisfies $\sum_{j \in B'} a_B(j) \leq v(B')$ for all bundles $B' \subseteq \items$, and $\sum_{j \in B'} a_B(j) = v(B')$ for $B' = B$. We will also require that $a_B(j) = 0$ for $j \not \in B$.

\item \textbf{Submodular valuations:} A  valuation function $v: 2^\items \rightarrow \Rp$ is submodular if  for any two item bundles  $A \subseteq B \subseteq \items$ and any item $j \in  \items \setminus B$, we have
$v(A \cup j) - v(A) \geq v(B \cup j) - v(B)$.

\item \textbf{Unit-demand valuations:} A unit-demand valuation function $v: 2^\items \rightarrow \Rp$ is such that for each bundle of items $B \subseteq \items$, we have $v(B) = \max_{j\in B} v(j)$. 
\end{itemize}

We will denote the contribution of item $j$ to the value of a bundle of items $B \subseteq \items$ given valuation function $v$ as $a_j(v, B)$, and define this to be the value assigned to $j$ by the additive supporting function of $v$ on $B$. 

\medskip\noindent\textbf{Prices, Utilities, and Demand.} For a vector of item prices $\bp = (p_1, \ldots, p_m) \in \Rp^m$ and a set of items $M \subseteq [m]$, we write $\demand(v,\bp,M)$ for the set of bundles $B \subseteq M$ that maximizes the utility $u^v(B,\bp) = v(B) - \sum_{j \in B} p_j$. We refer to $\demand(v,\bp,M)$ as the demand at valuation $v$ given prices $\bp$ when the set of available items is $M$. We use the shorthand $\demand(v,\bp)$ when $M = [m]$, or when $M$ is clear from the context. (Note that we can always simulate that some item $j \in [m]$ is not available by setting $p_j$ high enough.)

\medskip\noindent\textbf{Oracle Access.} Since any explicit description of a XOS valuation function $v: 2^\items \rightarrow \Rp$ would have size that is exponential in $m$, we consider three types of oracle access:
\begin{itemize}
\item \textbf{Value oracle:} A value oracle takes as input a valuation function $v$ and a bundle of items $B \subseteq \items$, and returns $v(B)$. 
\item \textbf{Demand oracle:} A demand oracle takes as input a valuation function $v$ and a vector of item prices $\bp = (p_1, \ldots, p_m) \in \Rp^m$, and returns a bundle of items $B \in \demand(v,\bp)$.
\item \textbf{XOS oracle (only defined for XOS valuation functions):} A XOS oracle takes as input a valuation function $v$ and a bundle of items $B \subseteq \items$, and returns the additive supporting function of $v$ on $B$.
\end{itemize}

\medskip\noindent\textbf{Posted-Price Mechanisms and Truthfulness.} We say that an online algorithm is a posted-price algorithm if for each buyer $i \in \buyers$, it proceeds as follows: Let $M_i \subseteq \items$ be the set of items that are still available when buyer $i$ arrives. The algorithm defines a vector of prices $\bp^{i} = (p^{i}_1,\ldots,p^{i}_m) \in \Rp^m$, and lets buyer $i$ choose a bundle  
of items $B \subseteq M_i$ that maximizes 
their utility,  
i.e., a set $B \in \demand(v_i,\bp^i,M_i)$.

Posted-price  algorithms where the vector of prices $\bp^{i}$ that bidder $i$ sees does not depend on buyer $i$'s valuation $v_i$ (but may depend on the valuations $v_{i'}$ of buyers $i' < i$, as well as on the distributions $D = (D_1, \ldots, D_n)$ of all buyers) are known to be dominant strategy incentive compatible (DSIC) --- or truthful \cite[e.g.,][]{FeldmanGL15}. Hence, they are also called sequential posted-price mechanisms.

%%%%%ONLINE COMBINATORIAL ALLOCATION WITH A SINGLE SAMPLE%%%%%
\section{Online Combinatorial Allocation with a Single Sample}

In this section we design $O(1)$-competitive algorithms for online combinatorial auctions with XOS buyers, given a single sample from the distributions of the bidders. We will in fact show a slightly stronger result, namely that this can be achieved in the Game of Googol model. I.e., our algorithm is a \emph{pointwise} $O(1)$-SPI.

\begin{theorem}[Restatement of \Cref{thm:mainSingleSample}]\label{thm:main-single-sample}
The problem of online combinatorial auctions with XOS buyers admits a pointwise $576$-competitive single-sample prophet inequality.
\end{theorem}

{  We remark that we did not try to optimize the constant factor in \Cref{thm:mainSingleSample}, and instead focused on the simplicity of exposition.}
Next, we first give an $O(1)$-competitive result assuming we have access to \emph{two samples}, 
and later we modify this algorithm to a single sample algorithm.

\subsection{Buyer-Wise Greedy}
We start with defining a key ingredient of our prophet inequality, the buyer-wise greedy algorithm (due to Dobzinski, Nisan, and Schapira~\cite{DobzinskiNS05}). 
This algorithm simply assigns every new buyer their demand set when the prices are given by the items' current contributions to the overall solution.

\smallskip

\begin{algo} 
\textbf{Buyer-Wise Greedy Algorithm.} We define $\GRD(\bv)$, for valuation profile $\bv$, as the assignment (and interchangeably, its value) determined as follows:

\begin{itemize}
    \item \edit{For all $i \in [n]$, define $\bp^{(i)} =(p_1^{(i)},\dots, p_m^{(i)})$ %for $j\in [m]$ and $i \in [n]$ 
    to be the prices faced by the i-th buyer.}\\
    Set $p_j^{(1)}=0$ for all $j\in [m]$.
    
    \item For each buyer $i=1,\dots, n$: (i.e., arriving in the fixed adversarial order)
    \begin{enumerate}
        \item Assign $i$ the bundle $\alloc_i= d_i(\bp^{(i)}) \in \demand(v_i,\bp^{(i)},[m])$, i.e., some demand set at prices $\bp^{(i)}$, with ties broken in an arbitrary but consistent manner. 
        (Re-assigning items in $\alloc_{i'} \cap \alloc_i$ from buyers $i' < i$ to buyer $i$.)
        \item Raise prices by setting $p_j^{(i+1)}= a_j(v_i, \alloc_i)$ for all $j \in \alloc_i$. %j\in [m]$.
        \\
        (Keeping $p_j^{(i+1)} = p_j^{(i)}$ for all $j \not \in \alloc_i$.)
    \end{enumerate}
\end{itemize}
\end{algo}

\smallskip

The following lemma of \cite{DobzinskiNS05} shows that the value attained by the buyer-wise greedy algorithm $\GRD(\bv)$ is  a $2$-approximation to $\OPT(\bv)$. 
For completeness, we provide a proof in Appendix~\ref{app:omitted-proofs}.

\begin{lemma}[Dobzinski, Nisan, and Schapira~\cite{DobzinskiNS05}]
\label{lem:greedy-approx}
    For \edit{any XOS} $\bv$, the output of the buyer-wise greedy algorithm $\GRD(\bv)$ is a $2$-approximation to the offline optimum $\OPT(\bv)$.
\end{lemma}

\noindent\textbf{Powers-of-Two Prices.}  
We modify the above procedure by making prices increase in an exponential fashion. 
Concretely, let 
\[\mathcal{P}=\left\{2^k\enspace|\enspace \edit{k \in \mathbb{Z}} \right\}.\]
Now, in the above greedy procedure $\GRD(\bv)$, we will raise the prices of items in the demand set $\alloc_i$ to the closest  power of two higher than $a_j(v_i, \alloc_i)$. 
\edit{That is, in Step 2 of the algorithm's main loop, rather than setting $p_j^{(i+1)} = a_j(v_i, \alloc_i)$, we instead set $p_j^{(i+1)} = \min\{ p \in \mathcal{P} \colon p \geq a_j(v_i, \alloc_i)\}$ for all $j \in \alloc_i$.}
Call the outcome of this modified procedure $\GRDTWO(\bv)$.

\begin{lemma}
\label{lem:modified-greedy-approx}
    For \edit{any XOS} $\bv$, the output of the modified buyer-wise greedy algorithm $\GRDTWO(\bv)$ where prices increase as powers of two is a $3$-approximation to the offline optimum $\OPT(\bv)$.
\end{lemma}

The proof of \Cref{lem:modified-greedy-approx} is a simple variation on the argument of \cite{DobzinskiNS05} and appears in \Cref{app:omitted-proofs}.

\medskip\noindent\textbf{Decoupling Allocation and Pricing.} 
We conclude by observing that we can decouple allocation and pricing as follows. Suppose that we have two valuations $f_i^1,f_i^2$ for each buyer $i \in [n]$ and that these are assigned to samples $\bs$ and valuations $\br$ by tossing an independent fair coin for each buyer. Let $\bp^{(i)}$ be the prices in the modified buyer-wise greedy run on $\bs$, and let buyer $i$ buy the set of items $\alloc_i$ that maximizes $r_i(\alloc_i) - \sum_{j \in B_i} p^{(i)}_j$. Note how this decouples allocation from pricing: while $\bs$ is used to set and update prices, $\br$ is used to make allocation decisions.

Analogously to \Cref{lem:modified-greedy-approx} we can now show the following lemma. Note that unlike \Cref{lem:modified-greedy-approx}, however, the following lemma does \emph{not} yet imply the existence of an $O(1)$-approximate allocation.  

\begin{lemma}\label{lem:modified-greedy-approx-with-samples}
Consider any pair of \edit{XOS} valuations $f_i^1,f_i^2$ for each buyer $i \in \buyers$ that are assigned independently and uniformly at random to $\bs,\br$. Then, for the modified buyer-wise greedy algorithm, it holds 
$\EXP{\sum_{i \in \buyers} r_i(B_i)} \geq \frac{1}{3} \EXP{\OPT(\bs)}.$
\end{lemma}

The proof (which appears in \cref{app:omitted-proofs}) exploits the way that $\bs$ and $\br$ are generated, to argue that in expectation allocations and prices align. So that the approximation guarantee can be established in a similar way as that in Lemma~\ref{lem:modified-greedy-approx}.

%--------------------------------------------------------------------------------------------------

\subsection{Two Samples Algorithm and Proof Overview}

We are now ready to prove the following theorem, which claims existence of an online procedure that attains the guarantee in Theorem~\ref{thm:main-single-sample} with two samples per buyer.

\begin{theorem}\label{thm:main-two-sample}
The problem of online combinatorial auctions with XOS buyers admits a pointwise $192$-competitive $2$-sample prophet inequality.
\end{theorem}

Recall what proving a statement like this entails: We need to consider a setting where an adversary writes down three valuation profiles $\bv^1, \bv^2, \bv^3$. Then, for each buyer $i \in \buyers$ independently, the valuation functions $v_i^1$, $v_i^2$, $v_i^3$ are assigned uniformly at random to samples $\bs$, $\bs'$ and actual valuations $\br$. 
The goal is to show that there is an online algorithm $\ALG$ whose expected value when receiving samples $\bs,\bs'$ and actual valuations $\br$ is a $O(1)$-approximation to the expected offline optimum $\OPT$ on the actual valuations $\br$.

\smallskip

Consider the following online procedure:

\begin{algo}
\textbf{Two Samples Algorithm.} Given two samples $\bs$ and $\bs'$, 
and arriving bidders with valuations $\br$:

  \begin{enumerate} 
    \item Run the modified buyer-wise greedy algorithm $\GRDTWO$ on the  sample $\bs$. For each buyer $i \in \buyers$, let $\bp^{(i)}$ denote the prices set by this algorithm on input $\bs_{<i}$ (i.e., on input $\bs$, but only considering the first $i-1$ buyers).\\
    \emph{(We emphasize that these prices increase by powers of $2$; this will be crucial for the max vs. sum argument in our analysis.)}
    
    \item For each buyer $i \in \buyers$ with valuation $s'_i\in \bs'$, determine their demand bundle $\alloc'_i$ given prices $\bp^{(i)}$.  Define \emph{base price}  $\base_j$ for any item $j \in [m]$  to be the largest contribution of $j$ to any  bundle $\alloc'_i$ with respect to $\sampleprime_i$. That is, $b_j := \max_{i: j \in \alloc'_i} a_j(s'_i,\alloc'_i)$.\\ 
    \emph{(Note that these base prices are not powers of $2$. In fact, we need all supporting prices to be distinct for a secretary-style argument.)}
    
    \item For each arriving buyer $i \in \buyers$ with valuation $r_i\in \br$:
    \begin{itemize}
        \item Compute their  demand set  $\bundle_i$  given the prices  $\bp^{(i)}$. 
        \item Assign them the still-available part of 
        $  B_i\cap \{j\in[m]\,|\,a_j(\realvalue_i,\bundle_i) > \base_j\}$.
    \end{itemize}
\end{enumerate}
Let us call the resulting allocation   (and interchangeably, its value) $\ALG$.\\

\end{algo}

\smallskip

Our goal is to now prove that $\ALG$ for any fixed valuation profiles $\bv^1, \bv^2, \bv^3$ in expectation over the coin tosses that lead to $\bs,\bs',\br$ is a constant-factor approximation to the offline optimum. Below, for an item $j \in \items$ assigned to buyer $i \in \buyers$ in $\ALG$, we use the shorthand $\ALG_j = a_j(r_i, B_i)$ to denote item $j$'s contribution to the algorithm's welfare. (Note that we define $\ALG_j$ with respect to the additive supporting function for bundle $\bundle_i$, so technically this is a lower bound on item $j$'s contribution to the welfare.)

\medskip\noindent\textbf{Proof Overview:} Our proof is broken into  two lemmas.  

The first lemma shows that the expected contribution of any item $j$ to the algorithm's welfare is  at least a constant fraction of $j$'s contribution to the max-bundle in $\br$ and $\bs'$. It is similar in spirit to arguments that appeared in \cite{AzarKW14} and \cite{RubinsteinWW20} for the single-choice problem. 

\begin{lemma} \label{lem:ALG_LB}
For any item $j \in \items$, 
\[ \EXP{\ALG_j} \geq \frac{1}{4} \cdot\EXP{\max\Bigg(\max_{i:\, j\in \bundle_i} a_j(r_i,\bundle_i),\, \max_{i:\, j\in \alloc'_i} a_j(s'_i,\alloc'_i)  \Bigg)} . \] 
\end{lemma}

Our second lemma shows that the contributions of any item $j$ to any demand bundle in the actual valuations $\br$ or the second sample $\bs'$ are dominated by just the largest of them (up to a constant factor).
This lemma and its proof is our key technical innovation.

\begin{lemma} \label{lem:MaxVsSum}
For any item $j \in \items$,
\[ 
% \forall j: 
\EXP{\max\Bigg(\max_{i:\, j\in B_i} a_j(\realvalue_i,\bundle_i),\, \max_{i:\, j\in \alloc'_i} a_j(\sample'_i,\alloc'_i)  \Bigg)}\geq \frac{1}{32}\cdot\EXP{\left(\sum_{i:\, j\in \bundle_i}a_j(\realvalue_i,\bundle_i) + \sum_{i:\, j\in \alloc'_i}a_j(\sample'_i,\alloc'_i)\right)}.
\]
\end{lemma}

Together,  the lemmas yield the constant competitive ratio, by summing up over all items $j$ and using the fact that the demand bundles are determined in exactly the same way as those in the modified buyer-wise greedy algorithm (which is a const. approx.).

\begin{proof}[Proof of \Cref{thm:main-two-sample}]
Summing over all items $j \in \items$ in  \Cref{lem:ALG_LB} and \Cref{lem:MaxVsSum} implies
\[ %\EXP{\ALG} = 
\sum_{j \in \items} \EXP{\ALG_j} ~\geq~ \frac{1}{128}\cdot\sum_{j \in \items}\EXP{\sum_{i:\, j\in \bundle_i}a_j(r_i,\bundle_i) + \sum_{i:\, j\in \alloc'_i} a_j(s'_i,\alloc'_i)} ~=~ \frac{1}{128} \cdot \EXP{ \sum_{i \in \buyers} \left(r_i(\bundle_i) + s'_i(\alloc'_i)\right)}.
\]
Note that in expectation the RHS is at least the sum of the expected welfare of $\GRDTWO(\br;\bs)$ and $\GRDTWO(\bs';\bs)$. 
Hence,  \Cref{lem:modified-greedy-approx-with-samples} implies that the RHS is at least $\frac{1}{128} \cdot \frac{2}{3} \cdot \EXP{\OPT(\br)}$, which proves the theorem.
\end{proof}

%--------------------------------------------------------------------------------------------------

\subsection{Proof of \Cref{lem:ALG_LB}: Secretary-Style Argument}
In this proof we will assume no ties in the additive supporting valuations, which can be ensured by adding noise. 
We prove the lemma for any fixed sample $\bs$, and in expectation over the assignment of the remaining two valuation functions $v_i^1,v_i^2$ for each buyer $i \in [n]$ to $\bs'$ and $\br$.

First fix sample $\bs$, which fixes all prices $\bp^{(i)}$. 
Note that this also fixes the demand bundles $B_i(\bs_{<i},v_i^k)$ for all $i\in [n]$ with respect to prices $\bp^{(i)}$ (set by $\bs_{<i}$) and valuations $v_i^k$ for $k\in \{1,2\}$. It also defines 
\[
a_j^{\max}:= \max_k\left(\max_{i:\, j\in B_i} a_j(v_i^k, B_i(\bs_{<i},v_i^k) )\right).
\]
Below we will show that with probability at least $\frac 14$ (where the probability is over the assignment of $v_i^1,v_i^2$ for $i \in \buyers$ to $\bs', \br$), the contribution of $j$ to the algorithm, $\ALG_j$, is at least $a_j^{\max}$.

Call the valuation that realizes $a_j^{\max}$, given it occurs at step $i^\star$, simply $v_{i^\star}^{\max}$.
Now we determine whether $v_{i^\star}^{\max}$ belongs to $\bs'$ or to $\br$ with a fair independent coin flip, therefore it will be in $\br$ with probability $1/2$.
Consider now the second-highest $j$-bid $a_j^{\text{2nd}}$. Either, it also belongs to buyer $i^\star$, in which case if $a_j^{\text{max}}\in \br$, it will belong to $\bs'$ with probability $1$. Or, it belongs to some other buyer $i \neq i^\star$, in which case another independent fair coin decides its belonging, and it will end up in $\bs'$ with probability $1/2$.
All in all,
$$\Pr[a_j^{\max}\in \br \text{ and } a_j^{\text{2nd}}\in \bs']\geq \frac 14.$$
Assume indeed both is the case. Then, $a_j^{\text{2nd}}$ will determine the cutoff/final price $\base_j$ of $j$, by definition. Also, by definition, among all $\br$-buyers attempting to obtain $j$, only $i^\star$ will be successful.
Therefore, we actually obtain value $a_j^{\max}$ for $j$. 
This yields
$\EXP{ \ALG_j }\geq \frac 14 a_j^{\max}$,
where the expectation is over the coin tosses that decide for each buyer $i \in \buyers$ the assignment of $v_i^1, v_i^2$ to $\bs', \br$, respectively.

%--------------------------------------------------------------------------------------------------

\subsection{Proof of \Cref{lem:MaxVsSum}: Handling Max vs. Sum}

We first compare the largest contribution of item $j$ in $\br$ with the final price of item $j$ as determined by $\bs$. Afterwards, we compare this final price to the sum of the $j$-bids in $\br$.
 
\begin{claim}
\label{cla:max-vs-price}
\[ \mathbb{E}\Big[\max_{i:\,j\in B_i}a_j(r_i,B_i)\Big] \geq \frac{1}{4}\mathbb{E}\Big[ p_j^{(n+1)} \Big]\]
i.e., we expect the largest $j$-bid in $\br$ to be not (too much) smaller than the greedy price of $j$ in $\bs$, after all $n$ $\bs$-buyers have been considered.
\end{claim}

\begin{proof}
To prove the claim, it suffices to show:
    \begin{align} \label{eq:probIneq}
        \forall \tau\geq 0: \enspace\enspace 2 \cdot \mathbb{P}\Big[ 2\max_{i:\,j\in B_i}a_j(r_i,B_i)\geq \tau\Big]  \geq \mathbb{P}\Big[ \max_{i\in\{1,\dots,n+1\}}p_j^{(i)} \geq \tau \Big] . 
    \end{align}
    This is because 
    \begin{align*}
    \mathbb{E}\Big[\max_{i:\,j\in B_i}a_j(r_i,B_i)\Big] ~&=~ \int_{\tau =0}^\infty \mathbb{P}\Big[\max_{i:\,j\in B_i}a_j(r_i,B_i) \geq \tau\Big] d\tau \\
    & = \frac12 \int_{\tau =0}^\infty \mathbb{P}\Big[2 \max_{i:\,j\in B_i}a_j(r_i,B_i)\geq \tau\Big] d\tau \\
    &\geq~ \frac14 \int_{\tau =0}^\infty \mathbb{P}\Big[ \max_{i\in\{1,\dots,n+1\}}p_j^{(i)} \geq \tau \Big] d\tau ~=~  \frac14 \mathbb{E}\Big[ \max_{i\in\{1,\dots,n+1\}}p_j^{(i)} \Big],
    \end{align*}
where the inequality uses Inequality~\eqref{eq:probIneq}.

To prove Inequality~\eqref{eq:probIneq}, suppose we have already decided the samples $\bs'$. We still need to determine the random assignment of the remaining two valuations $v_i^1,v_i^2$ for each buyer $i \in [n]$ to $\br$ and $\bs$, respectively. 
The decision tree for this stochastic process is a perfect binary tree with $n+1$ levels. In level $i$ we toss a fair coin and the two children at level $i+1$ correspond to the two possible assignments $r_i = v_i^1, s'_i = v_i^2$ resp.~$r_i = v_i^2, s'_i = v_i^1$. Note that there are $2^n$ root-leave paths, and that these are in one-to-one correspondence with the realizations of the coin tosses. 

We will now color (a subset of the vertices) in a top-down manner. Consider a node at level $i$. At this node we have already fixed $\br_{<i}$ and $\bs_{<i}$. We now decide the membership of $v_i^1,v_i^2$. Whenever one of $v_i^1,v_i^2$ has  has item $j$ in their demand bundle, we color the two children of the current node and all of its descendants
\begin{itemize}
    \item yellow, if, after the coin flip, $2 \max_{i' \leq i: j \in B_{i'}} a_j(r_{i'},B_{i'}) \geq \tau$.
    \item red, if, after the coin flip, $\max_{i'=1, \ldots,i+1} p^{(i')}_j \geq \tau$.
\end{itemize}

Note that a node can remain uncolored, may have only one of the two colors, or both.

The rationale is: We color a node yellow (resp.~red) when it becomes clear that on any root-leave path that passes through this node condition $2\max_{i:\,j\in B_i}a_j(r_i,B_i)\geq \tau$ (or condition $\max_{i\in\{1,\dots,n+1\}}p_j^{(i)} \geq \tau$) is met.   

To establish Inequality~\eqref{eq:probIneq} it thus suffices to argue that the number of yellow leaves is at least $1/2$ of the number of red leaves.
To this end, first consider a node at level $i$ that is not yet colored. The only way one of its children (and with them the respective subtrees) can become red is when at least one of $v_i^1,v_i^2$ has a $j$-bid of at least $\tau/2$. If both valuations have a $j$-bid of at least $\tau/2$, then both children (and their descendants) will have both colors. If only one of them has such a bid, then one subtree will be yellow and the other will be red.

So far this argument leads to the same number of yellow and red leaves. 
However, it is also possible that a node that is already colored gains a second color. Specifically, a node that is already yellow but not yet red can have one or both of its subtrees be colored red. If so, in the worst-case, (almost) all yellow leaves are also red.

Together this shows that the total number of red leaves is at most twice the number of yellow leaves, as claimed. This completes the proof. 
\end{proof}

\begin{claim}
\label{cla:price-vs-sum}
\[ \mathbb{E}\Big[p_j^{(n+1)}\Big] \geq \frac{1}{4}\mathbb{E}\Big[ \sum_{i\in [n]} a_j(r_i,B_i) \Big], \]
i.e., we expect the the final greedy price of item $j$ with prices set by $\bs$ to be not (too much) smaller then the sum of the $j$-bids in $\br$.
\end{claim}

\begin{proof}
    The proof of this claim is based on  a  Martingale approach.  Let
\[ X_i= \sum_{i'<i:\,j\in B_{i'}}a_j(r_{i'},B_{i'}) -4 \cdot  \max_{i'\in\{1,\dots,i\}}p_j^{(i')} . \]
Then, it suffices to show that $\mathbb{E}[X_{n+1}]  \leq 0$.  
Since $X_0 = 0$,  to prove  the claim  it suffices to show that the sequence $X_0,  X_1, \ldots, X_{n+1}$ forms a supermartingale, i.e.,  
\[\mathbb{E}\Big[X_{i}-X_{i-1} |X_1,\dots, X_{i-1}\Big]\leq 0.\]

To prove this,  consider the random step where we decide  between $r_i$ and $s_i$.  Either,  none of the two $j$-bids exceed the current price,  in which case $X_i = X_{i-1}$.  Otherwise,   at least one $j$-bid exists in the current step, the highest such bid $b$ is associated with $\br$ or $\bs$ with the same probability. If it ends up in $\bs$,  price $p_j^{(i+1)}$ will at least be twice of $p_j^{(i)}$,  and be at least the bid value $b$   afterwards.  Hence, $p_j$ increases by at least $b/2$.  On the other hand,  the sum always increases by at most $b$.  This gives  a supermartingale with a factor of $4$.
\end{proof}

Combining \Cref{cla:max-vs-price} and \Cref{cla:price-vs-sum}, we obtain
 \[\mathbb{E}\Big[\max_{i:\,j\in B_i}a_j(r_i,B_i)\Big] \geq \frac{1}{16}\mathbb{E}\Big[ \sum_{i\in [n]} a_j(r_i,B_i) \Big] .
 \]
 Since $ a_j(s'_i,A_i') $ are identically distributed as $a_j(r_i,B_i)$,  we similarly get
 \[\mathbb{E}\Big[\max_{i:\,j\in B_i}a_j(s_i',A_i')\Big] \geq \frac{1}{16}\mathbb{E} \Big[ \sum_{i\in [n]} a_j(s_i',A_i') \Big] .
 \]
Summing these two equations completes the proof of  \Cref{lem:MaxVsSum}.

%------------------------------------------------------------------------

\subsection{Extending to a Single Sample Algorithm}

We round up the proof of \Cref{thm:main-single-sample} by showing that we can turn any pointwise $k$-sample prophet inequality into a pointwise single-sample prophet inequality, only using a $O(k)$ term in the approximation guarantee.

\begin{theorem}
\label{thm:reduction}
    For any $\alpha$-competitive pointwise $k$-SPI with the objective of maximizing the sum over valuations of chosen elements in the single online steps, there exists an $\alpha (k+1)$-competitive pointwise  $1$-SPI. 
\end{theorem}

\begin{proof}
As input to the Game of Googol with $k =1$ we are given two valuations $f_i^1,f_i^2$ for each buyer $i \in \buyers$. For each buyer $i \in \buyers$, there is an independent fair coin toss to decide sample $s_i$ and actual valuation $r_i.$ The algorithm sees all the $s_i$ at the outset, while it can observe the $r_i$ only one-by-one in an online fashion.

We will show how to construct from this an input to the Game of Googol with $k > 1$. Recall that we need to construct, for each buyer, $k$ samples $\bar{s}_i^1, \ldots, \bar{s}_i^k$ and one actual valuation $\bar{r}_i$.  

Here's a first attempt at this: We take $s_i$ and $k$ zero functions, and assign them uniformly at random to $\bar{s}_i^1, \ldots, \bar{s}_i^k$ and $\bar{r}_i$. We feed  $\bar{s}_i^1, \ldots, \bar{s}_i^k$ and $\bar{r}_i$ into the algorithm $\ALG$ for the Game of Googol with $k > 1$. Then, the expected optimum in the resulting instance with $k$ samples is a $(k+1)$-approximation to that of the original instance of our $1$-sample game: consider any buyer $i$'s allocated bundle in the optimal solution on $\bs$, and their valuation $OPT_i(\bs)$ for it. With probability $1/(k+1)$ (independently), $s_i$ will be allocated to $\bar r_i$, implying $\EXP{\OPT(\bar\br)}\geq \frac{1}{k+1}\EXP{\OPT(\bs)}$ (and the same holds for the $OPT(\bar\bs_l)$, $l\in [k]$).

There is one issue with this construction though: This is completely ignoring the actual valuations $r_i$.
To make the construction useful in the online setting, 
whenever the coin flips were such that $\bar{r}_i = s_i$, we will substitute $s_i$ with $r_i$.
While in general, the input functions to the Game of Googol cannot depend on the coin flips made later to assign them, we will now justify that in our case, this is actually equivalent to our above first attempt.

Since all coin flips, those for the $1$-sample game as well as those for the $k$-sample game, are independent, we can imagine them to be drawn in the following order:
first, the assignments of functions to all $\bar s_i$ and $\bar r_i$ are decided.
Only then, we imagine the coins for the assignment of functions in the $1$-sample game to be drawn, i.e., each two input functions $f^1_i$, $f^2_i$ are assigned to $s_i$ or $r_i$, respectively. 
With this, for each $i$, exchanging $r_i$ for $s_i$ makes no difference. Both is simply the result (to be drawn later) of whatever the fair coin flip between $f^1_i$ and $f^2_i$ results in.

Since $\ALG$ is an $\alpha$-approximation to the Game of Googol with $k > 1$ samples, and the expected optimum of the instance that we feed into this algorithm is a $(k+1)$-approximation to the expected optimum of the original instance, we get an $\alpha(k+1)$ approximation as claimed.
\end{proof}

We now have all the ingredients to prove \Cref{thm:main-single-sample}.

\begin{proof}[Proof of \Cref{thm:main-single-sample}]
By \Cref{thm:main-two-sample} there is a pointwise $O(1)$-competitive two-sample prophet inequality, applying \Cref{thm:reduction} we can turn this into a pointwise single-sample prophet inequality while losing only a factor $2$ in the approximation guarantee.
\end{proof}

We conclude by noting that \Cref{thm:main-single-sample} together with Theorem~6.2 in \cite{CDFFLLP-SODA22} implies \Cref{cor:xos-oos}: the existence of a $O(1)$-competitive algorithm for online combinatorial auctions with XOS buyers in the order-oblivious secretary model.

%%%%%POSTED-PRICE MECHANISMS WITH POLYNOMIAL SAMPLES%%%%%
\section{Posted-Price Mechanisms with Polynomial Samples}

In this section we design $(2+\epsilon)$-competitive truthful combinatorial auctions \edit{for  bidders with independent XOS valuations} via sequential posted-pricing.  The prices we construct will have the additional property that each item sells with \edit{half} probability, over randomness in the agent valuations.  We will call such prices \emph{Median Prices}, defined formally below.  We show that, \edit{even for non-XOS valuations}, such prices always exist 
and are learnable with a polynomial number of samples that is independent of the valuation distributions, and in particular does not require \edit{bounded valuations}.
Finally, we show that the median property directly implies a constant approximation to welfare for XOS valuations.

\subsection{Generic Distributions and Median Prices}\label{sec:medianPrices}

Informally, median prices have the property that when agents make sequential purchase decisions, each item is sold with probability $1/2$.
An important subtlety is how agent indifferences are handled. 
We will first bypass these issues of tie-breaking by focusing on \emph{generic} valuation distributions,  a generalization of the atomless property, where for any fixed price vector $\bp$ indifferences occur with probability zero.  Then, in \Cref{sec:non-generic} we  relax this generic assumption and extend our results to general distributions.

\begin{definition}[Generic Distribution]
Distribution $D_i$ over valuations is \emph{generic} if for any price vector $\bp$ and any subset $M_i \subseteq [m]$, the demand correspondence of agent $i$ (given prices $\bp$ and set $M_i$ of available items) contains only a single set almost surely.  That is,  \[\PRO[v_i]{\exists B_1, B_2 \subseteq M_i \colon B_1 \neq B_2, \{B_1, B_2\} \subseteq \demand_i(v_i, \bp, M_i)} = 0 .\]
Distribution $D$ over valuation profiles is generic if its marginal distribution over each agent's valuation is generic.
\end{definition}

We interpret the genericness condition as saying that a buyer has zero probability of being indifferent between two bundles, and hence between whether or not to include a given item $j$ in their purchase.  This generalizes the atomless property to valuations.  We further note that one can make any distribution over valuations generic by adding an arbitrarily small noise to the agent valuations; see \Cref{sec:appendix-generic} for a formal construction.  For technical convenience we will first restrict attention to generic distributions, then in \Cref{sec:non-generic} we will show how to relax this assumption.

Given a fixed profile of generic valuation distributions, the probability that item $j$ is sold under price vector $\bp$ is well-defined.  We will write $\pi_j(\bp)$ for this probability, omitting the dependence on the distributions.
We are now ready to define what is meant by median prices.

\begin{definition}[Median Prices]
Given $\alpha \in [0, 1/2]$ and generic valuation profile distribution $D$, prices $\bp$ are $\alpha$-median prices if $\pi_j(\bp) \leq 1-\alpha$ for each $j$, and $\pi_j(\bp) \geq \alpha$ for each $j$ with $p_j > 0$.
\end{definition}

Note the interpretation of parameter $\alpha$: $(1/2)$-median prices have the property that each item is sold with probability exactly $1/2$, whereas $(1/4)$-median prices have the property that each item is sold with probability between $1/4$ and $3/4$.  Lower values of $\alpha$ are less restrictive, and every price vector $\bp$ is trivially $0$-median for any distribution $D$.  Also note that we permit an item with price $0$ to be sold with probability less than $\alpha$; this edge case ensures that median prices exist even if there is no demand for a given item even if it is free.

%--------------------------------------------------------------------------------------------------

\subsection{Truthful Combinatorial Auctions via Median Prices}
We now show that median prices can be used to construct approximately welfare-efficient posted pricing mechanisms.

\begin{theorem}
\label{thm:median-ppmech-approx-generic}
Let $D$ be a generic product distribution over XOS valuation profiles, and let $\bp$ be $\alpha$-median prices.  Then the expected social welfare of a sequential posted-price mechanism using prices $\bp$ is at least $\alpha$ times the expected optimal social welfare. 
\end{theorem}
\begin{proof}
    Let $OPT$ denote the optimal social welfare. We will also write $M_i(\bp, \bv)$ for the set of items still available for purchase for agent $i$ given the choice of prices and the agent valuations.  We will write $S(\bp, \bv)$ for the set of items that are allocated to any agent.  
    
    We now consider separately the expected revenue and the expected buyer surplus generated by the mechanism.

    The expected revenue generated by the mechanism is
    \begin{align*}
        \EXP[\bv]{\sum_j p_j \times \mathbbm{1}[j \in S(\bp, \bv)]}
        ~=~ \sum_j p_j \times \PRO[\bv]{j \in S(\bp, \bv)} 
         ~\geq~ \alpha \sum_j p_j,
    \end{align*}
    where we first use linearity of expectation and then that each item sells with probability at least $\alpha$.

    For buyer surplus, write $u_i$ for the random variable denoting the utility obtained by agent $i$.  Note that since each buyer chooses a utility-maximizing set, $u_i$ must be at least the utility obtained by agent $i$ from any set of available goods.  In particular, if we write $A^*(\bv)$ for the welfare-maximizing allocation for valuation profile $\bv$, we will consider a strategy of agent $i$ in which they consider drawing a virtual choice of valuations $\bv'_{-i}$ for the other agents, then purchase those items from $A^*(v_i, \bv'_{-i})$ that are still available and for which the prices $\bp$ are not too high.

    To bound the utility from this strategy we will make use of the supporting prices property of XOS valuations.  Fix any valuation profile $\bv$.  Since the valuations are XOS, there exist supporting prices $w_{ij}(\bv) \geq 0$ such that, for each $i$, we have $v_i(A^*_i(\bv)) = \sum_{j \in A^*_i(\bv)}w_{ij}(\bv)$, and moreover for any $T \subseteq A^*_i(\bv)$ we have $v_i(T) \geq \sum_{j \in T}w_{ij}(\bv)$.
    
    Given an agent $i$, prices $\bp$, valuation profile $\bv$, and a hallucinated valuation profile $\bv'$ of the other agents, we will define a subset $T_i(\bp, \bv, \bv'_{-i})$ of $A^*_i(v_i, \bv'_{-i})$ to be $A^*_i(v_i, \bv'_{-i}) \cap M_i(\bp, \bv) \cap \{ j \colon w_{ij}(\bv) \geq p_j\}$.   That is, $T_i(\bp, \bv, \bv'_{-i})$ is the subset of agent $i$'s optimal allocation (under the hallucinated profile of the other agents) that are available (under the true valuations) and for which the prices no larger than $w_{ij}(\bv)$.  We then have
    \begin{align*}
    \EXP[\bv]{u_i}
    & \geq \EXP[\bv, \bv'_{-i}]{ v_i \left( T_i(\bp, \bv, \bv'_{-i}) \right) - \sum_{j \in T_i(\bp, (v_i, \bv'_{-i})) }p_j } \\
    & \geq \EXP[\bv, \bv'_{-i}]{ \sum_{j \in T_i(\bp, \bv, \bv'_{-i}) } (w_{ij}(v_i, \bv'_{-i}) - p_j)} \\
    & = \EXP[v_i, \bv'_{-i}]{\sum_{j \in A^*_i(v_i, v'_{-i})} \left((w_{ij}(v_i, \bv'_{-i}) - p_j)^+ \times \PRO[\bv_{-i}]{j \in M_i(\bp, \bv)}\right)} \\
    & = \EXP[\bv]{ \sum_{j \in A^*_i(\bv)} \left( (w_{ij}(\bv) - p_j)^+ \times \PRO[\bv'_{-i}]{j \in M_i(\bp, (v_i, \bv'_{-i}))} \right)} \\
    & \geq \EXP[\bv]{\sum_{j \in A^*_i(\bv)} (w_{ij}(\bv) - p_j)^+ \times \alpha } .
    \end{align*}
    
    Taking a sum over all items, we conclude (writing $OPT$ for the expected optimal social welfare)
    \[ \EXP{\sum_i u_i} \geq \alpha\left(OPT - \sum_j p_j\right).\]

    Since the expected social welfare is the expected revenue plus the expected buyer surplus, we conclude that the expected welfare is at least
    \[ \alpha\sum_j p_j + \alpha(OPT - \sum_j p_j) = \alpha OPT ,\]
    as required.
\end{proof}

%--------------------------------------------------------------------------------------------------

\subsection{Existence of Median Prices}
\Cref{thm:median-ppmech-approx-generic} shows that median prices guarantee high welfare when used as posted prices.  But it is not immediately obvious that such prices are even guaranteed to exist.
We now prove the existence of $(1/2)$-median prices for generic distributions over \edit{(possibly non-XOS)} valuations.

\begin{theorem}
\label{theorem:medianpricesexistence}
For any generic distribution $D$ over valuations there exist $(1/2)$-median prices.
\end{theorem}

Fix a generic distribution $D$. In order to show the existence we will use the Kakutani fixed-point theorem on an appropriately defined function. Recall that $\pi_j(\bp)$ is the (well-defined) probability that item $j$ is allocated under prices $\bp$ and write $\pi(\bp)$ for the profile of allocation probabilities for the items.  We will first show that $\bp \mapsto \pi(\bp)$ is a continuous non-decreasing function.

\begin{lemma}
    The function $\pi(p_j, \bp_{-j})$ is continuous and, for each $j$, $\pi_j(p_j, \bp_{-j})$ is weakly decreasing in $p_j$.
\end{lemma}

\begin{proof}
Write $A_i(\bv, \bp)$ for the allocation to agent $i$ under prices $\bp$ when the valuations are $\bv$.

To see that it is weakly decreasing, take any $p'_j > p_j$ and any valuation $\bv$.  If the allocations $\mathbf{A}(\bv, \bp)$ and $\mathbf{A}(\bv, (p'_j, \bp_{-j}))$ are identical, then either both allocate $j$ or neither does.
Otherwise, take the smallest $i$ for which $A_i(\bv, \bp) \neq A_i(\bv, (p'_j, \bp_{-j}))$.  Then agent $i$ faces the same set of available items under $\bp$ and $(p'_j, \bp_{-j})$.  Since agent $i$ always receives a set from her demand correspondence, 
and agent $i$ faces the same set of available items and the same prices except for a higher price on item $j$, we must have $j \in A_i(\bv, \bp)$ and $j \not\in A_i(\bv, (p'_j, \bp_{-j}))$.  In particular, this means that if $\mathbf{A}(\bv, \bp)$ does not allocate item $j$, then neither does $\mathbf{A}(\bv, (p'_j, \bp_{-j}))$.  Integrating over all choices of $\bv$ then implies that the probability of allocating item $j$ is weakly lower under price $p'_j$ than under price $p_j$.

To show continuity of $\pi_j$, we proceed as follows. For any agent $i$ and sets of items $S$ and $M'$, let $f_{i, S, M'}(\bp)$ be the probability that under price vector $\bp$ buyer $i$ buys the set $S$ conditional on the set of remaining items being $M'$.\footnote{Note that if $D$ is not a product distribution, the conditioning on $M'$ impacts the conditional distribution over $v_i$.} It is now enough to show that for each $i$, $S$, and $M'$ such that $M'$ is the set of remaining items for $i$ with positive probability, the function $f_{i, S, M'}$ is continuous.  This is because we can write $\pi_j(\bp)$ as
\[
\pi_j(\bp) = \sum_{\substack{S_1, \ldots, S_n \subseteq [m]\\ j \in S_1 \cup \ldots \cup S_n}} \prod_{i=1}^n f_{i, S_i, [m] \setminus (S_1 \cup \ldots \cup S_{i-1})}(\bp).
\]

To show continuity of $f_{i, S, M'}$ for a fixed choice of $i$, $S$, and $M'$, let
\[
F(\bp, t) = \PRO{v_i(S) - \sum_{j \in S} p_j - \max_{S' \subseteq M', S \neq S'} \left( v_i(S') - \sum_{j \in S'} p_j \right) \leq t}.
\]
We observe that for all $\bp$, we have \[
F(\bp, t) - \lim_{\substack{t \to 0\\t<0}} F(\bp, t) = \PRO{v_i(S) - \sum_{j \in S} p_j - \max_{S' \subseteq M', S \neq S'} \left( v_i(S') - \sum_{j \in S'} p_j \right) = 0} = 0
\]
because the distribution is generic.\footnote{Here we used the assumption that $M'$ has non-zero probability of being the set of remaining items for agent $i$. If agent $i$ has non-zero probability of indifference between $S$ and $S'$ conditional on the set of remaining items being $M'$, then they have non-zero probability of this event unconditionally, contradicting genericness.} That is, the function $F(\bp, \cdot)$ is continuous at $0$. Furthermore, if $\lVert \bp - \bp' \rVert_\infty \leq \delta$, then $F(\bp, t - m \delta) \leq F(\bp', t) \leq F(\bp, t + m \delta)$. So, in order to show that $f_{i, S, M'}$ is continuous at $\bp$, let $\epsilon > 0$ and choose $\delta > 0$ such that $\lvert F(\bp, m x) - F(\bp, 0) \lvert < \epsilon$ for all $x \in [-\delta, \delta]$. Now, for any $\bp'$ with $\lVert \bp - \bp' \rVert_\infty \leq \delta$, we have $F(\bp', 0) \leq F(\bp, m \delta) < F(\bp, 0) + \epsilon$ and $F(\bp', 0) \geq F(\bp, - m \delta) > F(\bp, 0) - \epsilon$. The statement now follows because $f_{i, S, M'}(\bp) = F(\bp, 0)$ for all $\bp$.
\end{proof}

Using monotonicity and continuity of the allocation property, we can complete the proof of existence of $(1/2)$-median prices.

\begin{proof}[Proof of \Cref{theorem:medianpricesexistence}]
We will construct a mapping $\Psi$, which takes as input a vector of item prices and returns a vector of closed intervals of item prices.  Informally, we wish to define $\Psi(\mathbf{p})$ to be vector of intervals $\mathbf{Q}$ where each $Q_j$ is the closed interval of prices $[q^1_j, q^2_j]$ such that item $j$ sells with probability $1/2$ on any price vector $(\mathbf{p}_{-j}, q_j)$ with $q_j \in [q^1_j, q^2_j]$.  We'll then argue that $\Psi$ has a fixed point, and this will be our $1/2$-median prices.  

There are a few edge cases to consider when defining the mapping formally.  First, write $\bar{p}$ to be the median of $\max_i v_i(M)$.  Then it must be that $\pi_j(\bar{p}, \bp_{-j}) \leq 1/2$ for any $\bp$ because for item $j$ to sell at price $\bar{p}$ it is necessary that $v_i(M) \geq \bar{p}$ for an $i$.  We will focus attention on price vectors lying in $[0,\bar{p}]^m$.  Fix input $\mathbf{p}$, and define $\phi_j(q_j) := \pi_j(q_j, \mathbf{p}_{-j})$.  
We have two cases.  Either $\phi_j(0) < 1/2$, or $\phi_j(0) \geq 1/2$.  In the first case, we will define $\Psi_j(\mathbf{p}) = \{0\}$.  In the latter case, must exist some $q > 0$ at which $\phi_j(q) = 1/2$, by continuity and the fact that $\phi_j(\bar{p}) \leq 1/2$.  Note there may be more than one such $q$, so write $Q_j$ for the set of $q$ that satisfy this property, which must be a closed interval since $\phi_j$ is continuous and weakly decreasing.  Then we will define  $\Psi_j(\mathbf{p}) = Q_j$.  

This concludes the definition of $\Psi$.  We now note that $\Psi$ is a set-valued mapping from a compact space (i.e., $[0,\bar{p}]^m$) to itself with closed values.  Moreover, $\Psi$ is upper hemi-continuous.  To see why, consider some sequence of prices vectors $\bp^k$ that converge to $\bp$ as $k \to \infty$, with a corresponding sequence $\mathbf{q}^k$ converging to $\mathbf{q}$.  We then have that item $j$ sells with probability $1/2$ on price vector $(q^k, \bp^k_{-j})$ for all $k$.  Since $(q_j^k, \bp^k_{-j})$ converges to $(q_j, \bp_{-j})$, we conclude by continuity of $\pi$ that $j$ sells with probability $1/2$ on price vector $(q_j, \bp_{-j})$ for each $j$, and hence $\mathbf{q} \in \Psi(\bp)$.

Since $\Psi$ is upper hemi-continuous, by the Kakutani fixed-point theorem it has a fixed point $\mathbf{p}$.  That is, there exists $\mathbf{p}$ such that $p_j \in \Psi_j(\mathbf{p})$ for all $j$. This fixed point is a choice of $1/2$-median prices, by definition.
\end{proof}

%--------------------------------------------------------------------------------------------------
\subsection{Learning Median Prices}
Next, we can also find median prices with polynomially many samples. Note that in contrast to similar learnability results for prices defined on the mean (e.g., in \cite{FeldmanGL15}), our error is only multiplicative. That is, we do not need the distribution of values to be bounded.

\begin{theorem}
\label{theorem:medianpricesfromsamples}
    For any  generic distribution $D$, with probability $1 - \delta$, $O\left(\frac{1}{\epsilon^2}\left(m^2 + m \log n + \log(1/\delta)\right)\right)$ samples suffice to find $(\frac{1}{2}-\epsilon)$-median prices.
\end{theorem}

Suppose we are given $k$ vectors of $n$ valuation functions that are drawn from the unknown distribution. Construct the empirical distribution $\hat{D}$, in which  one of these $k$ vectors is drawn uniformly at random. (Note that there are correlations between the agents in this empirical distribution.)

For draws from the generic distribution $D$, regardless of the price vector $\bp$, almost surely no tie occurs. This is not true for the empirical distribution $\hat{D}$. In the following, we will fix any deterministic tie-breaking rule that does not depend on prices or the set of remaining items. That is, every agent has a permutation of the power set of the set of items and whenever there is a tie between $S$ and $S'$ agent $i$ will prefer one over the other.

Based on this tie-breaking, we let $\hat{\pi}_j(\bp)$ be analogously the probability that item $j$ gets sold under prices $\bp$ if a valuation profile $\bv$ is drawn from $\hat{D}$. Note that it is safe to assume that the same tie-breaking rule is applied to define the probability $\pi_j(\bp)$ that item $j$ gets sold under price vector $\bp$ if the valuation profile $\bv$ is drawn from $D$. 

The following lemma states a uniform convergence property that with high probability, the sale probability of each item given \emph{any} price vector is approximately equal between the empirical distribution and the true distribution.

\begin{lemma}[Uniform Convergence]
\label{lem:convergence}
Consider the empirical distribution over $k$ sampled valuation profiles, with $k \geq C \frac{1}{\epsilon^2}\left(m^2 + m \log n + \log(1/\delta)\right)$ for a constant $C$.
Then, with probability at least $1 - \delta$, we have $\lVert \hat{\pi}(\bp) - \pi(\bp)\rVert_\infty < \epsilon$, that is, there is no price vector on which the probability that an items sells differs by more than $\epsilon$ between the actual and the empirical distribution.
\end{lemma}

Before proving the lemma, we note that this implies \Cref{theorem:medianpricesfromsamples}.  We know that there exists a $(1/2)$-median price vector $\bp^*$ for $D$.  So with probability $1-\delta$, $\bp^*$ is $(1/2-\epsilon)$-median for $\hat{D}$, and hence in particular a $(1/2-\epsilon)$-median price profile exists.  Let $\bp$ be \emph{any} $(1/2-\epsilon)$-median price profile for $\hat{D}$.  Then, again conditioning on the event of probability $1-\delta$, $\bp$ will be $(1/2 - 2\epsilon)$-median for $D$.  Taking $k$ samples therefore suffices to obtain $(1/2 - 2\epsilon)$-median prices with probability $1-\delta$.

\begin{proof}[Proof of \Cref{lem:convergence}]
Consider the space of all possible valuation vectors $\bv$.  Every price vector $\bp$ defines a subset $S_j^{\bp}$ of this space, equal to the collection of valuation vectors for which item $j$ is sold under price vector $\bp$. By this definition $\pi_j(\bp) = \PRO[\bv \sim D]{\bv \in S_j^{\bp}}$ and $\hat{\pi}_j(\bp) = \PRO[\bv \sim \hat{D}]{\bv \in S_j^{\bp}}$.

Standard bounds on uniform convergence (e.g., in \cite{wainwright_2019} combine Theorem 4.10 and Equation (5.50)) imply that with $k$ samples, we have for any $\gamma>0$,
\begin{align*}
    \PRO{\sup_{\bp} \lvert \hat{\pi}_j(\bp) - \pi_j(\bp) \rvert \geq C \sqrt{\frac{\nu}{k}} + \gamma} \leq \exp(-k\gamma^2/2),
\end{align*}
where $C$ is a constant and $\nu$ is the VC dimension of the set system $\{ S_j^{\bp} \}$. Below, in \Cref{lem:vc}, we show that the VC dimension of the set system defined by all of these subsets is at most $O(m^2 + m \log n)$. 

So for any item $j$ and any $\epsilon, \delta' > 0$, this implies that $O\left(\frac{\nu + \ln(1/\delta')}{\epsilon^2}\right)$ samples are sufficient for $\PRO{\sup_{\bp} \lvert \hat{\pi}_j(\bp) - \pi_j(\bp) \rvert \geq \epsilon} \leq \delta'$.

By furthermore taking a union bound with $\delta' = \frac{\delta}{m}$ over all items $j$, $O\left(\frac{\nu + \ln(1/\delta)}{\epsilon^2}\right)$ samples are sufficient for $\PRO{\sup_{\bp} \lVert \hat{\pi}(\bp) - \pi(\bp) \rVert_\infty \geq \epsilon} \leq \delta$, meaning that with probability at least $1-\delta$, we have that $\lvert \hat{\pi}_j(\bp) - \pi_j(\bp) \rvert < \epsilon$ for all items $j$ and all price vectors $\bp$.
\end{proof}

So, it remains to bound the VC dimension of the set system $\{ S_j^{\bp} \}$ defined above. We use the technique by Balcan et al.~\cite{BalcanSV18} (see also \cite{BalcanDDKSV21}). 

\begin{lemma}
\label{lem:vc}
The VC dimension of the set system $\{ S_j^{\bp} \}$ is at most $O(m^2 + m \log n)$.
\end{lemma}

\begin{proof}
Let $\bv^{(1)}, \ldots, \bv^{(N)}$ be $N$ valuation profiles that are shattered. That is, for every set $T \subseteq [N]$, there is a price vector $\bp^T$ such that item $j$ gets sold on valuation profile $v^{(t)}$ if and only if $t \in T$. We will give an upper bound on the number of price vectors that lead to different outcomes on $\bv{v}^{(1)}, \ldots, \bv{v}^{(N)}$.

Note that if on a valuation profile $\bv$, item $j$ gets sold under price vector $\bp$ but doesn't get sold under price vector $\bp'$, there have to be $i$, $S$, and $S'$ such that agent $i$ prefers $S$ under prices $\bp$ but $S'$ under prices $\bp'$ because otherwise the outcomes would be identical. That is, $v_i(S) - \sum_{j' \in S} p_{j'} \geq v_i(S') - \sum_{j' \in S'} p_{j'}$ but $v_i(S) - \sum_{j' \in S} p'_{j'} \leq v_i(S') - \sum_{j' \in S'} p'_{j'}$, where one of the two inequalities is strict because the tie-breaking is independent of the prices. So, $v_i(S) - v_i(S') - \sum_{j' \in S \setminus S'} p_{j'} + \sum_{j' \in S' \setminus S} p_{j'} \geq 0$ but $v_i(S) - v_i(S') - \sum_{j' \in S \setminus S'} p'_{j'} + \sum_{j' \in S' \setminus S} p'_{j'} \leq 0$, where again one of the two inequalities is strict.

In other words, if we consider the space $\Rp^m$ of all price vectors, the price vectors $\bp$ and $\bp'$ lie on different sides of the hyperplane defined by $v_i(S) - v_i(S') - \sum_{j' \in S \setminus S'} p^\ast_{j'} + \sum_{j' \in S' \setminus S} p^\ast_{j'} = 0$.

The valuation profiles $\bv^{(1)}, \ldots, \bv^{(N)}$ define $N \cdot n \cdot (2^m)^2$ such hyperplanes. Due to a result by Buck \cite{Buck43}, this means there are no more than $m (N \cdot n \cdot (2^m)^2)^m$ different regions. If $\bv^{(1)}, \ldots, \bv^{(N)}$ is shattered, each of the $2^N$ price vectors $\bp^T$ for $T \subseteq [N]$ has to lie in a different region.

We conclude that $2^N \leq m (N \cdot n \cdot (2^m)^2)^m$.  Taking logs on both sides yields $N \leq \log m + m\log N + m\log n + 2m^2$, which implies $N = O(m^2 + m\log n)$.%\tknote{Double-check this.}
\end{proof}

Generally, the polynomial dependence in \Cref{lem:vc} is unavoidable.  Indeed, for the special case of unit-demand valuations, even with $n=1$ the VC dimension of the proposed set system is $\Omega(m)$.  Given $m$ items, and any $t$ with $1 \leq t < m$, let $v_1^{(t)}$ be unit-demand function that has value $1$ for item $1$, value $2$ for item $t+1$ and value $0$ for all other items. Consider price vectors with $p_1 = 0$. Now, depending on the price of item $t+1$, $v_1^{(t)}$ will prefer item $1$ (if $p_{t+1} < 1$) or item $t+1$ (if $p_{t+1} > 1$).  Thus, for every subset of the valuations $v_1^{(1)}, \ldots, v_1^{(m-1)}$, there is a price vector such that exactly this subset buys item $1$.

%--------------------------------------------------------------------------------------------------

\subsection{Computational Challenges}

The uniform convergence result from the previous section also provides us a very powerful tool to compute median prices, namely by considering the empirical distribution of polynomially many samples---which means that one has to argue only about polynomially many scenarios.

To demonstrate this strength, we now describe a Tatonnement style polynomial-time algorithm for computing median prices for unit demand valuations. We leave it as a challenging but interesting open problem to extend this result to XOS valuations.

\begin{theorem}
    For any generic distribution $D$ over unit demand valuations and any $\epsilon > 0$, there is a fully polynomial-time algorithm that, given polynomially many samples from $D$, generates a $(1/2-\epsilon)$-median distribution over item prices with probability $1-\delta$.  The runtime is polynomial in $n$, $m$, $1/\epsilon$, and $\log(1/\delta)$.
\end{theorem}
\begin{proof}
    Fix generic distribution profile $D$.  Our algorithm will begin by drawing $k$ sampled valuation profiles $(\bv^1, \dotsc, \bv^k)$ from $D$; let $\hat{D}$ denote the empirical distribution.  By \Cref{lem:convergence}, the probability of sale between $D$ and $\hat{D}$ agree within $\epsilon$ on all price vectors, with probability $1-\delta$.  We will condition on this event for the remainder of this proof. 
    
    Write $\hat{\pi}(\bp)$ for the profile of probabilities that each item sells under prices $\bp$ given the empirical distribution $\hat{D}$, assuming an arbitrary tie-breaking rule.  Note that $\hat{\pi}_j(\bp)$ is always a multiple of $1/k$, and can equivalently be thought of as counting how many of the $k$ samples result in item $j$ being purchased.  

    Given any price vector $\bp$, we will define a potential value $f(\bp)$ that depends on the probabilities of allocation.  This will be the sum of potential values $f_j(\bp)$ for each item $j$.  Assume for convenience that $k$ is even.  If $\hat{\pi}_j(\mathbf{p}) \leq 1/2$ then $f_j(\bp) = 0$.  If $1/2 < \hat{\pi}_j(\mathbf{p}) \leq 1/2 + \epsilon$, say $\hat{\pi}_j(\mathbf{p}) = 1/2 + t/k$ where $t \leq \epsilon k$, then $f_j(\bp) = t^2$.  If $\hat{\pi}_j(\mathbf{p}) > 1/2 + \epsilon$, say $\hat{\pi}_j(\bp) = 1/2 + t/k$ where $t > \epsilon k$, then $f_j(\bp) = (\epsilon k)t$.  Note then that $0 \leq f(\bp) \leq \epsilon k^2 m$ for all prices $\bp$.

    We are now ready to describe our algorithm.
    We will initialize all item prices $p_j$ to $0$.  
    
    If $\hat{\pi}_j(\mathbf{p}) \leq 1/2 + \epsilon$ for all $j$ then we are done.  Otherwise, since $k/\epsilon > m$, there must exist some multiple of $1/k$, say $t/k$, such that $t/k \in (1/2, 1/2+\epsilon)$ and $\hat{\pi}_j(\bp) \neq t/k$ for all $j$.  We will take $S$ to be the set of all $j$ such that $\hat{\pi}_j(\bp) > t/k$.  Note that $S$ is non-empty, since in particular any item $j$ with $\hat{\pi}_j(\bp) > 1/2 + \epsilon$ must be in $S$.

    We will then raise the price of each item in $S$, uniformly, until some agent $i$ in some sample $k$ switches from demanding an item in $S$ to demanding an item in $[m]-S$ (or demanding no item).  Note that since the valuations are unit-demand, these are the only changes to the demand of an agent that can occur.

    The amount we need to raise the price for this to occur is the minimum, over all $i$, $\ell$, and $j \in S$ and $j' \not\in S$, of $(v_i^\ell(j) - p_j) - (v_i^\ell(j') - p_{j'})$. (Where here we can allow $j' = 0$ to represent taking the empty set, with value and price both equal to $0$.) We can therefore find this uniform price increment in polynomial time by enumerate each of these polynomially-many options.

    Under this update, the allocation probability of some item $j$ with $\hat{\pi}_j(\mathbf{p}) > t/k$ reduces by $1/k$, and the allocation probability of some item $j'$ with $\hat{\pi}_j(\mathbf{p}) < t/k$ increases by $1/k$.  Thus, if we write $\bp'$ for the new price vector, we conclude that $f(\bp) - f(\bp') \geq (t+1)^2 + (t-1)^2 - 2t^2 = 2$.

    We then repeat, choosing a new set $S$ and iterating, terminating if $\hat{\pi}_j(\bp) \geq 1/2 + \epsilon$ for all $j$.  Note that if we do not terminate, then in particular we have $\hat{\pi}_j(\bp) > 1/2 + \epsilon$ for some $j$, so $f(\bp) \geq f_j(\bp) > \epsilon k^2$.  On the other hand, since $f(\bp)$ is at most $\epsilon k^2 m$ initially and reduces by at least $2$ on each iteration, we conclude that the procedure must terminate after at most $\epsilon k^2 (m-1)/2$ iterations.  This algorithm therefore finds the desired price vector in polynomially many iterations, as claimed.
\end{proof}

\subsection{Beyond Generic Distributions} \label{sec:non-generic}

To this point in our discussion of median prices we have focused on the case of generic distributions. We now show how to remove this assumption.  We will show how to define median prices to non-generic distributions, then show that our welfare bound in Theorem~\ref{thm:median-ppmech-approx-generic}, existence result in Theorem~\ref{theorem:medianpricesexistence}, and learning result in Theorem~\ref{theorem:medianpricesfromsamples} all extend to general non-generic settings.

For non-generic distributions, an agent $i$ may be indifferent between multiple sets with positive probability.  The probability with which an item sells therefore depends on how agents choose between sets in their demand correspondence.
It will thus be convenient to define, for each agent, a choice rule that describes how indifferences are to be broken.

A \emph{choice rule} $A_i$ for agent $i$ is a mapping from a demand correspondence $\demand \subseteq 2^{[m]}$ (which recall is a set of sets of items) to a distribution over elements of $\demand$.  We think of $A_i(\demand)$ as the (possibly randomized) set chosen by agent $i$ in the face of indifference.  Given a profile of choice rules $\mathbf{A} = (A_1, \dotsc, A_n)$, a valuation profile $\bv$, and prices $\bp$, we will write $\mathbf{A}(\bv,\bp)$ for the (random) allocation that occurs when each agent $i$ sequentially selects from their demand correspondence (given $v_i$, $\bp$, and the set of remaining items $M_i$) according to $A_i$.  We will also write $A_i(\bv,\bp)$ for the (random) allocation to agent $i$ under allocation profile $\mathbf{A}(\bv,\bp)$.

\begin{definition}[Median Prices for General Distributions]
Given $\alpha \in [0, 1/2]$ and any profile $D$ of distributions of bidder valuations, we say that prices $\bp$ are $\alpha$-median prices if there exists a choice rule profile $\mathbf{A}$ such that, when agents acquire bundles according to $\mathbf{A}$, the probability that item $j$ is allocated lies in $[\alpha, 1 - \alpha]$ for every $j$.
\end{definition}

Note that prices $\bp$ are median prices if there exists an appropriate corresponding choice rule $\mathbf{A}$.  We say that such a choice rule is a \emph{witness} for prices $\bp$.

\begin{example}
Suppose there is a single item $\{1\}$ and two agents $\{1,2\}$.  Agent values are deterministic (i.e., $D$ is a point mass): agent $1$ has value $3$ and agent $2$ has value $4$.  Then the price $p_1 = 4$ is a median price, witnessed by any choice rule such that $A_1$ is arbitrary and $A_2(\{\{1\}, \emptyset\})$ randomizes uniformly between $\{1\}$ and $\emptyset$.
\end{example}

\medskip\noindent\textbf{Welfare of Median Prices for General Distributions.}
We are now ready to extend Theorem~\ref{thm:median-ppmech-approx-generic}, our approximation result for median prices, to non-generic settings. Indeed, under our definition of median prices for general distributions, the proof of Theorem~\ref{thm:median-ppmech-approx-generic} proceeds unchanged.  We require only that in our sequential posted-price mechanism, indifferences are broken according to a choice rule $\mathbf{A}$ that witnesses the median prices.  The reason is that genericness is never used explicitly in the proof; we use only that each item is purchased with probability lying in $[\alpha, 1-\alpha]$ under the posted-price mechanism, which is guaranteed if the witnessing choice rule is followed.  We obtain the following result.

\begin{theorem}
\label{thm:median-ppmech-approx-nongeneric}
Let $D$ be a product distribution over XOS valuation profiles, and let $\bp$ be $\alpha$-median prices witnessed by choice rule $\mathbf{A}$.  Then the expected social welfare of a sequential posted-price mechanism using prices $\bp$ and choice rule $\mathbf{A}$ to resolve indifferences is at least $\alpha$ times the expected optimal social welfare. 
\end{theorem}

\medskip\noindent\textbf{Existence of Median Prices for General Distributions.} 
Next, we will extend Theorem~\ref{theorem:medianpricesexistence} and show that $(1/2)$-median prices exist even for non-generic distributions \edit{over (possibly non-XOS) valuations}.  The idea behind the proof is to consider perturbations of the non-generic distribution, yielding generic distributions for which median prices exist.  We then take a limit of vanishing perturbations to find median prices for the original distribution.  To find the witnessing choice rule, we take an appropriate limit of allocations over the sequence of perturbations.  We obtain the following result.

\begin{theorem}
\label{theorem:medianpricesexistence-nongeneric}
For any \edit{(not necessarily generic)} distribution $D$ over valuation profiles there exist $(1/2)$-median prices.
\end{theorem}
\begin{proof}
Following the outline above, we will construct a perturbed distribution and consider its median prices.
Given distribution $D$ 
over valuation profiles and $\epsilon > 0$, write $D^{(\epsilon)}$ for the profile of generic distributions from the proof of \Cref{prop:distrib.noise.generic}. Recall that, to draw a valuations $\bv'$ from distribution $D^{(\epsilon)}$, we first draw $\bv$ from $D$, then draw an additive value $w_{ij}$ uniformly from $[0,\epsilon]$ for each agent $i$ and item $j$, and then finally set $v'_i(S) = v_i + \sum_{j \in S}w_{ij}$ for each $S \subseteq [m]$.  In \Cref{prop:distrib.noise.generic} we prove that $D^{(\epsilon)}$ is generic for any $\epsilon > 0$, 
so by Theorem~\ref{theorem:medianpricesexistence} it admits $(1/2)$-median prices.  Write $\bp^{(\epsilon)}$ for these prices.

First, we argue that the price vectors $\bp^{(\epsilon)}$ are uniformly bounded.
Consider the optimal welfare obtainable from a valuation profile drawn from $D$.  This is a random variable; let $\mu$ denote its highest median value, which must be finite.  If $\mu^{(\epsilon)}$ is the corresponding median for distribution $D^{(\epsilon)}$, then note we must have $\mu^{(\epsilon)} \leq \mu + m\epsilon$ (as the optimal welfare can increase by at most $m\epsilon$ pointwise under the addition of $\epsilon$-bounded random noise for each item).  So $\mu^{(\epsilon)}$ is at most a finite constant depending on $D$, for any $0 < \epsilon < 1$.  Finally, note that we must have $p_j^{(\epsilon)} \leq \mu^{(\epsilon)}$ for all $j$, as the probability that any set sells at a price higher than $\mu^{(\epsilon)}$ is less than $1/2$ by definition.  We therefore conclude that $p^{(\epsilon)}$ lies in $[0,\mu+m]^m$ for all $0 < \epsilon < 1$.

Since the prices $\bp^{(\epsilon)}$ lie in a compact set, the sequence $\bp^{(\epsilon)}$ as $\epsilon\rightarrow 0$ must have a convergent subsequence.  Write $\bp^*$ for the limit of this convergent subsequence.

We claim that $\bp^*$ is a profile of $(1/2)$-median prices for valuation profile distribution $D$.  To prove this claim we must construct a witnessing choice rule $\mathbf{A}^*$.
We begin by constructing a choice rule $\mathbf{A}^{(\epsilon)}$ for each $\epsilon > 0$.\footnote{These will not be witnesses for prices $\bp^{(\epsilon)}$; they will simply be helpful for defining $\mathbf{A}^*$.}
For any $\epsilon > 0$, any agent $i$, and any collection of sets $\demand$, we will define 
\[ A_i^{(\epsilon)}(\demand) \in \arg\max_{S \in \demand}\left\{ \sum_{j \in S}\left(w_j + p_j^* - p_j^{(\epsilon)}\right)\right\} \]
where each $w_j$ is drawn uniformly from $[0,\epsilon]$ 
independently for each $j$.  Note that if the argmax is not unique in the definition of $A_i^{(\epsilon)}(\demand)$ then we can choose an element in the argmax arbitrarily, as this happens with probability 0 for generic distribution $D^{(\epsilon)}$.

To give some intuition into the definition of $A_i^{(\epsilon)}$, note that for any $v_i$ and any set $S$, if we write $v'_i$ for the perturbed version of $v_i$ in the definition of $D_i^{(\epsilon)}$, then $\sum_{j \in S}\left(w_j + p_j^* - p_j^{(\epsilon)}\right) = (v'_i(S) - \sum_{j \in S}p_j^{(\epsilon)}) - (v_i(S) - \sum_{j \in S}p_j^*)$.  In other words, $A_i^{(\epsilon)}$ breaks indifferences according to preference under the perturbed valuations and the corresponding median prices.

As the space of choice rules is compact (mapping a discrete set to probability distributions over a discrete set), there must be a subsequence of $\mathbf{A}^{(\epsilon)}$ as $\epsilon \to 0$ that converges with respect to total variation distance.  Let $\mathbf{A}^*$ denote the limit, which must be a choice rule.  We claim that $\mathbf{A}^*$ is a witnessing choice rule for $\bp^*$.

To prove the claim, fix some $\epsilon$ and a valuation profile $\bv$ drawn from $D$.  For each $i$ we will couple $D_i^{(\epsilon)}$ and $A_i^{(\epsilon)}$ by using the same uniform weights $w_j \sim U[0,\epsilon]$ in the definition of each.  Write $v'_i$ for the resulting perturbed valuation from $D_i^{(\epsilon)}$.  Let $(S_1, \dotsc, S_n)$ be the sets allocated to the agents under valuations $\bv$, prices $\bp$, and choice rule $A_i^{(\epsilon)}$, with each $S_i$ being chosen from a corresponding demand set $\demand_i$.  Let $(T_1, \dotsc, T_n)$ be the sets allocated under perturbed valuations $\bv'$ and prices $\bp^{(\epsilon)}$.  Since $D^{(\epsilon)}$ is generic, each $T_i$ uniquely maximizes $u_i^{v'_i}(T,\bp^{(\epsilon)})$ almost surely, so we condition on that event.  
Our goal now is to show that for all $i$ and all sufficiently small $\epsilon$, $S_i = T_i$ with high probability.

Suppose there is some smallest $i$ such that $S_i \neq T_i$.
Set $T_i$ uniquely maximizes $u_i^{v'_i}(T_i,\bp^{(\epsilon)})$ over all items available to agent $i$.  On the other hand, $S_i \in \arg\max_{S \in \demand_i}\{\sum_{j \in S}(w_j + p_j^* - p_j^{(\epsilon)}\} = \arg\max_{S \in \demand_i}\{u^{v'_i}(S, \bp^{(\epsilon)})\}$.  Thus $T_i \neq S_i$ implies $T_i \not\in \demand_i$.  So in particular we must have $u_i^{v_i}(T_i,\bp^*) < u_i^{v_i}(S_i,\bp^*)$ 
but $u_i^{v'_i}(T_i,\bp^{(\epsilon)}) > u_i^{v'_i}(T_i,\bp^{(\epsilon)})$.  Since $v'_i$ and $v_i$ differ by at most $m\epsilon$ on the value of any set, we conclude that $T_i \neq S_i$ implies
\begin{align}
\label{eq:util.cond1}
u_i^{v_i}(T_i, \bp^*) < u_i^{v_i}(S_i, \bp^*) < u_i^{v_i}(T_i, \bp^*) + m\epsilon + \sum_j|p_j^* - p_j^{(\epsilon)}|.
\end{align}
Our strategy will be to argue that \eqref{eq:util.cond1} can occur only with small probability.

For any $\gamma > 0$, $i \in [m]$, and $S,T \subseteq [m]$, let $E(\gamma, i, S, T)$ denote the event that, for prices $\bp^*$ and valuation $v_i$ drawn from $D_i$, $0 < u_i^{v_i}(S, \bp^*) - u_i^{v_i}(T, \bp^*) \leq \gamma$. We claim that $\lim_{\gamma \to 0} \PRO{E(\gamma, i, S, T)} = 0$ for any distribution over the values $v_i(S)$ and $v_i(T)$. To see this, let $F(x) = \PRO{u_i^{v_i}(S, \bp^*) - u_i^{v_i}(T, \bp^*) \leq x}$. We have $\PRO{E(\gamma, i, S, T)} = F(\gamma) - F(0)$ and, as $F$ is the CDF of a random variable, it is right-side continuous, so $\lim_{\gamma \to 0, \gamma > 0} F(\gamma) = F(0)$.

Thus, if we define event $E(\gamma)$ to be $\cup_{i,S,T}E(\gamma,i,S,T)$, a union bound implies that $\lim_{\gamma \to 0} E(\gamma) = 0$ as well.  So for any $\delta > 0$ there exists $\gamma_\delta > 0$ such that $\PRO{E(\gamma_\delta)} < \delta$. In particular, with probability at least $1-\delta$, the loss in utility from any non-demanded set taken by any agent is at least $\gamma_\delta$.

Since $\lim_{\epsilon \to 0}\bp^{(\epsilon)} = \bp^*$, 
There exists some $\epsilon_1 > 0$ such that for all $\epsilon < \epsilon_1$ and all $j$, $\sum_j|p^*_j - p_j^{(\epsilon)}| < \gamma_\delta/2$. We conclude that for all $\epsilon < \min\{\epsilon_1, \gamma_\delta/(2m)\}$, if event $E(\gamma')$ does not occur then inequality \eqref{eq:util.cond1} cannot hold, since if $u_i^{v_i}(T_i, \bp^*) < u_i^{v_i}(S_i, \bp^*)$ then, from the definition of $E(\gamma')$,
\begin{align*} 
u_i^{v_i}(S_i, \bp^*) & > u_i^{v_i}(T_i, \bp^*) + \gamma_\delta\\
& = u_i^{v_i}(T_i, \bp^*) + \gamma_\delta/2 + \gamma_\delta/2\\ 
& > u_i^{v_i}(T_i, \bp^*) + m\epsilon + \sum_j|p_j^* - p_j^{(\epsilon)}|.
\end{align*}

We conclude that if event $E(\gamma_\delta)$ does not occur, then for all sufficiently small $\epsilon$ it must be that $S_i = T_i$ for all $i$, and hence the distribution over allocations (and in particular the probability that each item $j$ sells) is identical.  Some notation: write $\pi_j(D, \bp, \mathbf{A})$ for the probability of sale of item $j$ under valuation distributions $D$, prices $\bp$, and choice rule $\mathbf{A}$, where we can omit $\mathbf{A}$ from the notation for generic distributions.  Then since $\PRO{E(\gamma_\delta)} < \delta$, we conclude that $|\pi_j(D, \bp^*, \mathbf{A}^{(\epsilon)}) - \pi_j(D^{(\epsilon)}, \bp^{(\epsilon)})| < \delta$ for all $j$ and all sufficiently small $\epsilon$.  
As $\mathbf{A}^{(\epsilon)}$ converges to $\mathbf{A}^*$ in total variation distance, we conclude that 
$|\pi_j(D, \bp^*, \mathbf{A}^*) - \pi_j(D^{(\epsilon)}, \bp^{(\epsilon)})| < 2\delta$ for all sufficiently small $\epsilon$.
Since $\bp^{(\epsilon)}$ are median prices for $D^{(\epsilon)}$, taking $\delta \to 0$ completes the proof.
\end{proof}

\medskip\noindent\textbf{Learning Median Prices for General Distributions.}
Finally, we will show that median prices not only exist for general distributions, we can also find them with polynomially many samples.

\begin{theorem}
\label{theorem:medianpricesfromsamples-nongeneric}
    Let $D$ be any (not necessarily generic) distribution over valuation profiles. With probability $1 - \delta$, $O\left(\frac{1}{\epsilon^2}\left(m^2 + m \log n + \log(1/\delta)\right)\right)$ samples suffice to find $(\frac{1}{2}-\epsilon)$-median prices and their witnessing choice rule.
\end{theorem}

We will consider choice rules that are determined by a tie-breaking vector $\bq \in \R^m$. The choice rule $\mathbf{A}^{\bq}$ is defined as follows. Each agent $i$ draws $w_{i, 1}, \ldots, w_{i, m} \stackrel{\text{i.i.d.}}{\sim} U[0,1]$ and then chooses whatever $S$ in their demand correspondence maximizes $\sum_{j \in S} (w_j - q_j)$ (any remaining ties broken arbitrarily in a fixed way, say lexicographically). A consequence of the proof of \Cref{theorem:medianpricesexistence-nongeneric} is that for every $\gamma > 0$ there are $(\frac{1}{2}-\gamma)$-median prices $\bp^*$ that are witnessed by $\mathbf{A}^{\bq^*}$ for some $\bq^* \in \R^m$, namely setting $q_j = \frac{p_j^* - p_j^{(\epsilon')}}{\epsilon'}$ the choice rules $\mathbf{A}^{\bq}$ and $\mathbf{A}^{(\epsilon')}$ are equivalent. Now consider $\epsilon'$ that is small enough so that the probability of sale for every item $j$ under prices $\bp^*$ and $\mathbf{A}^{(\epsilon')}$ are within $[\frac{1}{2} - \gamma, \frac{1}{2} + \gamma]$.

For any choice of a price vector $\bp$ and tie-breaking vector $\bq$, we let $\pi_j(\bp, \bq)$ be the probability that item $j$ is sold under $\bp$ with choice rule $\mathbf{A}^\bq$ being applied. Our goal is to to find $\bp$ and $\bq$ based on samples from $D$ such that $\pi_j(\bp, \bq)$ is close to $\frac{1}{2}$.

Let us call $(v_i, w_i)$ the \emph{extended} valuation of agent $i$. For a fixed price vector $\bp \in \Rp^m$ and a fixed tie-breaking vector $\bq \in \R^m$, the extended valuation profile $(\bv, \bw)$ completely defines the outcome on $\bp$ is choice rule $\mathbf{A}^{\bq}$ is applied.

Consider again the empirical distribution defined by $k$ drawn for extended valuation profiles $(\bv, \bw)$. Note that $k$ draws from $D$ can easily be extended to draws of extended valuation profiles by appending $w_{i,j}$ drawn i.i.d.\ from $U[0,1]$. For the empirical distribution of extended valuation profiles, we let $\hat{\pi}_j(\bp, \bq)$ be the fraction of extended valuation profiles such that item $j$ gets sold under prices $\bp$ with tie-breaking $\bq$. We again have a uniform-convergence guarantee.

\begin{lemma}[Uniform Convergence]
\label{lem:convergence-nongeneric}
Consider the empirical distribution over $k$ sampled extended valuation profiles, with $k \geq C \frac{1}{\epsilon^2}\left(m^2 + m \log n + \log(1/\delta)\right)$ for a constant $C$.
Then, with probability at least $1 - \delta$, we have $\lVert \hat{\pi}(\bp, \bq) - \pi(\bp, \bq)\rVert_\infty < \epsilon$.
\end{lemma}

This bound implies \Cref{theorem:medianpricesfromsamples-nongeneric}.  We know that there exists a $(1/2-\epsilon)$-median price vector $\bp^*$ for $D$ witnessed by $\mathbf{A}^{\bq^*}$ for some $\bq^*$.  So with probability $1-\delta$, $\pi_j(\bp^*, \bq^*) \in [\frac{1}{2} - 2 \epsilon, \frac{1}{2} + 2 \epsilon]$ for all $j$.  Let $\bp$ and $\bq$ be \emph{any} vectors such that $\hat{\pi}_j(\bp, \bq) \in [\frac{1}{2} - 2 \epsilon, \frac{1}{2} + 2 \epsilon]$ for all $j$.  Then, again conditioning on the event of probability $1-\delta$, $\hat{\pi}_j(\bp, \bq) \in [\frac{1}{2} - 3 \epsilon, \frac{1}{2} + 3 \epsilon]$ for all $j$.

\begin{proof}[Proof of \Cref{lem:convergence-nongeneric}]
Consider the space of all possible extended valuation vectors $(\bv, \bw)$.  Every pair of a price vector $\bp$ and tie-breaking vector $\bq$ defines a subset $S_j^{\bp, \bq}$ of this space, equal to the collection of valuation vectors for which item $j$ is sold under price vector $\bp$. By this definition $\pi_j(\bp, \bq) = \PRO[\bv \sim D, \bw \sim U^m{[0,1]}]{(\bv, \bw) \in S_j^{\bp, \bq}}$ and $\hat{\pi}_j(\bp, \bq) = \PRO[\bv \sim \hat{D}, \bw \sim U^m{[0,1]}]{(\bv, \bw) \in S_j^{\bp}}$.

Let $\nu$ be the VC dimension of the set system $\{ S_j^{\bp, \bq} \}$. We will show below in \Cref{lem:vc-nongeneric} that also the VC dimension of the set system defined by all of these subsets is upper bounded by $O(m^2 + m \log n)$.

Therefore, also here, for any item $j$ and any $\epsilon, \delta' > 0$, this implies that $O\left(\frac{\nu + \ln(1/\delta')}{\epsilon^2}\right)$ samples are sufficient for $\PRO{\sup_{\bp, \bq} \lvert \hat{\pi}_j(\bp, \bq) - \pi_j(\bp, \bq) \rvert \geq \epsilon} \leq \delta'$.

Again, taking a union bound with $\delta' = \frac{\delta}{m}$ over all items $j$, $O\left(\frac{\nu + \ln(1/\delta)}{\epsilon^2}\right)$ samples are sufficient for $\PRO{\sup_{\bp, \bq} \lVert \hat{\pi}(\bp, \bq) - \pi(\bp, \bq) \rVert_\infty \geq \epsilon} \leq \delta$, meaning that with probability at least $1-\delta$, we have that $\lvert \hat{\pi}_j(\bp, \bq) - \pi_j(\bp, \bq) \rvert < \epsilon$ for all items $j$ and all price vectors $\bp$ and all tie-breaking vectors $\bq$.
\end{proof}

The bound on the VC dimension works in a similar way as before.

\begin{lemma}
\label{lem:vc-nongeneric}
The VC dimension of the set system $\{ S_j^{\bp, \bq} \}$ is at most $O(m^2 + m \log n)$.
\end{lemma}

\begin{proof}
Let $(\bv^{(1)}, \bw^{(1)}), \ldots, (\bv^{(N)}, \bw^{(1)})$ be $N$ extended valuation profiles that are shattered. That is, for every set $T \subseteq [N]$, there is a price vector $\bp^T$ and a tie-breaking vector $\bq^T$ such that item $j$ gets sold on $(v^{(t)}, w^{(t)})$ if and only if $t \in T$.

Note that if on an extended valuation profile $(\bv, \bw)$, item $j$ gets sold under $(\bp, \bq)$ but doesn't get sold under price vector $(\bp', \bq')$, there have to be $i$, $S$, and $S'$ such that agent $i$ prefers $S$ under $(\bp, \bq)$ but $S'$ under $(\bp', \bq')$ because otherwise the outcomes would be identical. That is, $v_i(S) - \sum_{j' \in S} p_{j'} \geq v_i(S') - \sum_{j' \in S'} p_{j'}$ but $v_i(S) - \sum_{j' \in S} p'_{j'} \leq v_i(S') - \sum_{j' \in S'} p'_{j'}$ and one of the two inequalities is strict or $\sum_{j' \in S} (w_{i, j'} - p_{j'}) \geq \sum_{j' \in S'} (w_{i, j'} - p_{j'})$ but $\sum_{j' \in S} w_{i, j'} - p'_{j'}) \leq \sum_{j' \in S'} (w_{i, j} - p'_{j'})$ and one of the two inequalities is strict.

That is, if we consider the $2m$-dimensional space $\Rp^m \times \R^m$ of all price vectors and tie-breaking vectors, the vectors $(\bp, \bq)$ and $(\bp', \bq')$ either lie on different sides of the hyperplane defined by $v_i(S) - v_i(S') - \sum_{j' \in S \setminus S'} p^\ast_{j'} + \sum_{j' \in S' \setminus S} p^\ast_{j'} = 0$ or they lie both on the hyperplane but on different sides of the one defined by $\sum_{j' \in S} w_{i, j} - \sum_{j' \in S'} w_{i, j} - \sum_{j' \in S \setminus S'} q^\ast_{j'} + \sum_{j' \in S' \setminus S} q^\ast_{j'} = 0$.

The extended valuation profiles $(\bv^{(1)}, \bw^{(1)}), \ldots, (\bv^{(N)}, \bw^{(N)})$ define $N \cdot n \cdot (2^m)^2$ such hyperplanes of each kind, so $2 \cdot N \cdot n \cdot (2^m)^2$ hyperplane in total. This means there are no more than $2 m (2 \cdot N \cdot n \cdot (2^m)^2)^{2m}$ different regions in $\R^{2m}$. If $(\bv^{(1)}, \bw^{(1)}), \ldots, (\bv^{(N)}, \bw^{(N)})$ is shattered, each of the $2^N$ vectors $(\bp^T, \bq^T)$ for $T \subseteq [N]$ has to lie in a different region.

We conclude that $2^N \leq 2 m (N \cdot n \cdot (2^m)^2)^{2m}$.  Taking logs on both sides yields $N \leq \log m + m\log N + m\log n + 2m^2$, which implies $N = O(m^2 + m\log n)$.
\end{proof}

\section{Open Problems and Future Directions}

% \begin{enumerate}
%     \item Truthful online combinatorial auction with  a single sample?
%     \item Polytime truthful online combinatorial auction with polynomial samples?
% \end{enumerate}

{ In this work we showed that a single sample from each bidder's value distribution is sufficient to obtain an $O(1)$-approximation for the online combinatorial allocation problem with XOS bidders, while a polynomial number of samples per bidder suffices to get a $(2+\epsilon)$-approximate truthful mechanism. An exciting frontier opened by our work is whether the problem with XOS bidders admits a \emph{truthful} $O(1)$-approximate mechanism with a poly-logarithmic number of samples, or even only a single sample from each bidder's value distribution. Another natural next step, would be to explore whether the online combinatorial allocation problem with \emph{subadditive} bidders can be solved with poly-many samples, or even with a single sample. }

%\clearpage

\bigskip

\appendix

%%%%%OMITTED PROOFS%%%%%
\section{Omitted Proofs}\label{app:omitted-proofs}

\subsection{Proof of \Cref{lem:greedy-approx}}

\begin{proof}[Proof of \Cref{lem:greedy-approx}]
First observe that the prices in the buyer-wise greedy algorithm are non-decreasing. This can be shown by contradiction. If  we had $a_j(v_i,\alloc_i) < p_j^{(i)}$ for some $i \in \buyers$ and $j \in \alloc_i$, then buyer $i$ would strictly prefer bundle $\alloc_i\setminus\{j\}$ over $\alloc_i$, in contradiction to $\alloc_i = d_i(\bp^{(i)}) \in \demand(v_i,\bp^{(i)},\items)$. Indeed, if we had $a_j(v_i,\alloc_i) < p_j^{(i)}$, then
\[
v_i(\alloc_i\setminus \{j\}) - \sum_{j \in \alloc_i\setminus \{j\}} p_j^{(i)} \geq \sum_{j \in \alloc_i\setminus \{j\}} \left(a_j(v_i,A_i) -  p_j^{(i)} \right) > \sum_{j \in \alloc_i} \left(a_j(v_i,\alloc_i) -  p_j^{(i)} \right)
= v_i(\alloc_i) - \sum_{j \in \alloc_i} p_j^{(i)},\]
where we used that $a(v_i,\alloc_i)$ is an additive supporting function for the first and last step, and the strict inequality comes from the assumption.

Using this, we now show the approximation guarantee. We will first lower bound 
$\GRD(\bv) \geq \sum_{j \in \items} p_j^{(n+1)}$. Afterwards, we will show that $\OPT(\bv) \leq \sum_{j \in \items} p_j^{(n+1)} + \sum_{i\in \buyers} u_i(A_i,\bp^{(i)})$ and that $\sum_{i\in \buyers} u_i(A_i,\bp^{(i)}) \leq \GRD(\bv)$. Together this shows that $\OPT(\bv) \leq 2 \GRD(\bv)$, as claimed.

Let's begin with the lower bound on $\GRD(\bv)$. 
Since $\GRD_i(\bv) \subseteq \alloc_i$ for all $i \in \buyers$, $a_j(v_i,A_i) = p_j^{(n+1)}$ for all $i \in \buyers$ and $j \in \GRD_i(\bv)$, and $p_j^{(n+1)} = 0$ for all $j \in \items$ that remain unassigned, we obtain
\begin{align}
\GRD(\bv) = \sum_{i \in \buyers} v_i(\GRD_i(\bv)) \geq \sum_{i \in \buyers} \sum_{j \in \GRD_i(\bv)} p_j^{(n+1)} = \sum_{j \in \items} p_j^{(n+1)}.
\label{eq:bound-on-gv}
\end{align}

For the upper bound on $\OPT(\bv)$, we proceed in two steps. 
First, observe that for each $i \in \buyers$, since $A_i = d(\bp^{(i)}) \in \demand(v_i,\bp^{(i)},\items)$, and prices are monotonically increasing, it holds that
\begin{align}
v_i(\OPT_i(\bv)) - \sum_{j \in \OPT_i(\bv)} p_j^{(n+1)} \leq v_i(\OPT_i(\bv)) - \sum_{j \in \OPT_i(\bv)} p_j^{(i)} \leq v_i(A_i) - \sum_{j \in A_i} p_j^{(i)}.
\label{eq:lower-opt}
\end{align}

Note that  the right-hand side $v_i(A_i) - \sum_{j \in A_i} p_j^{(i)} = u_i(A_i,\bp^{(i)})$. We now use that in each step of the algorithm, the sum of the additive supporting valuations increases exactly by the updating buyer's utility. Using this, we obtain
\begin{align}
\GRD(\bv) \geq \sum_{i \in \buyers} \sum_{j \in \GRD_i(\bv)} a_j(v_i,A_i) = \sum_{i = 1}^{n} \left(v_i(A_i) - \sum_{j \in A_i} p_j^{(i)} \right) 
= \sum_{i \in \buyers} \bigg(v_i(A_i) - \sum_{j \in A_i} p_j^{(i)}\bigg).
\label{eq:declared-utility}
\end{align}

By summing Inequality \eqref{eq:lower-opt} over all $i \in \buyers$ and using Equality \eqref{eq:declared-utility}, we obtain
\[
\OPT(\bv)- \sum_{j \in \items} p_j^{(n+1)} 
\leq \sum_{i \in \buyers} \bigg( v_i(A_i) - \sum_{j \in A_i} p_j^{(i)} \bigg) \leq \GRD(\bv).
\]
Rearranging this, shows $\OPT(\bv) \leq \sum_{j \in \items} p_j^{(n+1)} + \GRD(\bv)$, as claimed.
\end{proof}

\subsection{Proof of \Cref{lem:modified-greedy-approx}}

\begin{proof}[Proof of \Cref{lem:modified-greedy-approx}]
The proof of the approximation guarantee follows from the following observations. First note that as in the unmodified algorithm, prices are monotonically increasing. We can then use that $a_j(v_i,A_i) \geq p^{(n+1)}_j/2$ to conclude, similar to Inequality~\eqref{eq:bound-on-gv}, that
$\GRDTWO(\bv) \geq \frac{1}{2} \sum_{j \in \items} p^{(n+1)}_j.$ 
Moreover, the arguments that lead to Inequality~\eqref{eq:lower-opt} also apply to the modified algorithm, so that
\begin{align}\label{eq:bound-on-opt}
\OPT(\bv) \leq \sum_{j \in \items} p_j^{(n+1)} + \sum_{i \in \buyers} u_i(A_i,\bp^{(i)})
\end{align}
The argument is completed by observing that in the modified algorithm, although the sum of the additive supporting valuations no longer increase by exactly the updating buyer's utility, it still grows by \emph{at least} this amount. With this we can proceed as in Inequality~\eqref{eq:declared-utility}, except that the first equality becomes an inequality, and conclude that 
\begin{align}
\GRDTWO(\bv) = \sum_{i \in \buyers} \sum_{j \in \GRDTWO(\bv)} a_j(v_i,A_i) \geq \sum_{i \in \buyers} u_i(A_i,\bp^{(i)}).
\end{align}
Putting everything together, we obtain $\OPT(\bv) \leq \sum_{j \in \items} p^{(n+1)}_j + \sum_{i \in \buyers} u_i(A_i,\bp^{(i)}) \leq 3 \GRDTWO(\bv)$ as claimed.
\end{proof}

\subsection{Proof of \Cref{lem:modified-greedy-approx-with-samples}}

\begin{proof}[Proof of \Cref{lem:modified-greedy-approx-with-samples}]
Denote the demand set of agent $i$ with respect to valuation $r_i$ when prices $\bp^{(i)}$ are set through $\bs_{<i}$ by $B_i$; and let's write $A_i$ for the demand set of agent $i$ with respect to valuation $s_i$ when prices $\bp^{(i)}$ are set through $\bs_{<i}$. Since $\br$ and $\bs$ are generated one-by-one by flipping an independent fair coin, for all $i\in \buyers$, it holds that
$\EXP{r_i(B_i)} = \EXP{s_i(A_i)}$ and $\EXP{\sum_{j \in B_i} p_j^{(i)}} = \EXP{\sum_{j \in A_i} p_j^{(i)}}$.

Summing this over all $i \in \buyers$ yields
\begin{align*}
\EXP{ \sum_{i \in \buyers} r_i(B_i) } &\geq \EXP{ \sum_{i \in \buyers}\left( r_i(B_i) - \sum_{j \in B_i} p_j^{(i)}\right) } \\
&= \EXP{\sum_{i \in \buyers} \left(s_i(A_i) - \sum_{j \in A_i} p_j^{(i)}\right)} \geq \EXP{\OPT(s) - \sum_{j \in [m]} p_j^{(n+1)}},
\end{align*}
where the inequality holds by the same arguments as those that lead to Inequality \eqref{eq:bound-on-opt}.

Moreover, we have that
\[
\EXP{\sum_{j \in \items} p_j^{(n+1)}} = \EXP{\sum_{i \in \buyers} \sum_{j \in A_i} \left(p_j^{(i+1)} - p_j^{(i)} \right)} \leq \EXP{\sum_{i \in \buyers} \sum_{j \in A_i} p_j^{(i+1)}} \leq \EXP{\sum_{i \in \buyers} 2 s_i(A_i)}.
\]

We conclude that
\[
\EXP{\sum_{i \in \buyers} r_i(B_i)} + \EXP{\sum_{i \in \buyers} 2 s_i(A_i)} \geq \EXP{\OPT(\bs)},
\]
which shows that $3 \EXP{\sum_{i \in \buyers} r_i(B_i)} \geq \EXP{\OPT(\bs)}$ as claimed.
\end{proof}

%%%%%MAKING DISTRIBUTIONS GENERIC%%%%%
\section{Making Distributions Generic}
\label{sec:appendix-generic}

\begin{proposition}
\label{prop:distrib.noise.generic}
Given any distribution $D$ over valuations and any $\epsilon > 0$, there exists a generic distribution $D'$ and a coupling $F$ between $D$ and $D'$ such that for all $(v,v') \sim F$, $|v(S) - v'(S)| < \epsilon$ for all $S \subseteq [m]$. 
\end{proposition}
\begin{proof}
%\begin{claim}
Given any distribution $D$ over a class of valuations, consider the following distribution $D'$.  To draw $v' \sim D'$, first draw $v \sim D$, then draw $w_j \sim U[0,\epsilon]$ for each item $j$.  Then we define $v'(S) = v(S) + \sum_{j \in S}w_j$ for each $S \subseteq [m]$.

    We need to prove that $D'$ is generic.  Choose any price vector $\bp$ and any $M \subseteq [m]$.  Choose any $S,T \subseteq M$ with $S \neq T$.  We will show that the probability that $S$ and $T$ both lie in $\demand(v', \bp, M)$ is $0$.  
    To prove the claim, 
    note that in order for both $S$ and $T$ to be in $\demand(v', \bp, M)$, a necessary condition is that  $v'(S) - \sum_{j \in S}p_j = v'(T) - \sum_{j \in T}p_j$.  
    From the definition of $v'$, this occurs only if
    \[ \sum_{j \in S-T} w_j - \sum_{j \in T-S} w_j = (v(T)-v(S)) + (\sum_{j \in S}p_j - \sum_{j \in S}p_j).\]  
    Once we draw $v \sim F$ but before we draw the $w_j$ terms, the right-hand side is a fixed constant and the left-hand side is a random variable with no atoms (since $S \neq T$).  So the probability of equality is $0$.  Taking a union bound over all $S$ and $T$ yields the desired result.
\end{proof}

\section{Examples for Simple Algorithms} \label{sec:simpleAlgos}

We give examples demonstrating the difficulties arising for ‘simple' sample-based algorithms (for both the allocation and the auction problem), outlining the need for our more elaborate approach.

\medskip\noindent\textbf{The Natural Single-Choice Extension.}
Probably the most natural idea is to directly extend the single-choice SSPI that simply sets the largest observed sample value as a threshold.
As a first idea, one might attempt to accordingly set a price for each item that equals $\max_{v_s\in S}v_s(j)$. However, imagine an instance where every buyer has an additive value of $1$ for each item, except for buyer $1$, who has a value of $2$ for each nonempty subset of the items, i.e. also for each singleton: then, the algorithm never sells more than one item, resulting in a ratio of $\Omega(m)$.

\medskip\noindent\textbf{Item Prices as Contribution to Offline Optimum.}
Consider the following sequential posted pricing algorithm based on a single sample.
Given a single sample, compute the offline optimum allocation for that sample, then price each item according to its supporting price under that allocation (capturing its ``contribution'' to said optimum).
This approach doesn't work.  Suppose $n=m$ and consider a distribution $D$ over profiles of unit-demand valuations, such that player $n$ always has value $1$ for item $1$ and value $\epsilon$ for every other item.  Each other player and item has value $\epsilon/2$.  Then assuming $\epsilon \ll 1/n$, the optimum for any sample assigns item $1$ to player $1$ for a total value of $1 + O(n\epsilon)$.  So this approach sets $p_1 = 1$ and $p_j = \epsilon/2$ for all $j \neq 1$.  But under these prices, player $1$ never chooses item $1$, so the total value generated is $O(n\epsilon)$ (since a truthful mechanism must assign buyers their favourite bundle under the given prices). Without truthfulness, the algorithm could instead e.g. assign each player the largest-valued item that is still price-feasible. While this approach seems promising, one needs to find a way of showing that such an algorithm will not sell out high-valued items before the according buyers arrive.

\medskip\noindent\textbf{Balanced Prices from a Single Sample.} Motivated by balanced pricing, suppose we instead take a single sample, compute an offline optimal allocation, then set each item's price to be half of its supporting price.  This also does not work.  Suppose $n=m$ and consider a distribution $D$ over profiles of unit-demand valuations, such that player $n$ has value $1$ for one of the items (uniformly at random), and value $\epsilon$ for the remaining items.  Each other player has value $\epsilon$ for every item. Then this pricing approach will result in price $1/2$ for one of the items (uniformly at random), say item $j$, and price $\epsilon/2$ for all other items.  In the real valuation, the optimal allocation has value $1 + O(n\epsilon)$ by assigning to player $n$ whichever item is valued at $1$ (say item $i$).  But unless $j = i$, player $n$ will receive no items: all items other than $j$ will be taken by the first $n-1$ players, leaving only item $j$ which has price $1/2 > \epsilon$.  So the total value obtained will be $O(n\epsilon)$.

The construction above assumed a deterministic order in which agents choose their items.  Note that if we modify the example by taking $n \gg m$, then the analysis would extend to (a) an algorithm where the agent order is randomized for sequential posted pricing, and even (b) scenarios where there isn't deterministically a single agent with high value for an item, but all agent valuations are drawn iid from a distribution with probability $O(1/n)$ of having the ``high'' type (value $1$ for some item uniformly at random).

In conclusion, it is not possible to correctly estimate the value our algorithms ``should'' obtain from an item to compute prices. Instead, our algorithms' key property is using the fact that both samples and actual valuation stem from the same distribution to bound the probability that items are prematurely assigned.

\bigskip

%%%%%ACKNOWLEDGMENTS%%%%%
\section*{Acknowledgements}
We are thankful to Shuchi Chawla, Trung Dang, and anonymous reviewers 
for comments on earlier versions of the paper. 
The last author is thankful to Ellen Vitercik for explaining \cite{BalcanDDKSV21}.

\clearpage

{\small
\bibliographystyle{alpha}
%\bibliography{bib}

\begin{thebibliography}{GHT{\etalchar{+}}21}

\bibitem[AKS21]{AKS-SODA21}
Sepehr Assadi, Thomas Kesselheim, and Sahil Singla.
\newblock Improved truthful mechanisms for subadditive combinatorial auctions:
  Breaking the logarithmic barrier.
\newblock In {\em Proceedings of {SODA}}, pages 653--661, 2021.

\bibitem[AKW14]{AzarKW14}
Pablo~Daniel Azar, Robert Kleinberg, and S.~Matthew Weinberg.
\newblock Prophet inequalities with limited information.
\newblock In {\em Proceedings of SODA}, pages 1358--1377, 2014.

\bibitem[AS19]{AS-FOCS19}
Sepehr Assadi and Sahil Singla.
\newblock Improved truthful mechanisms for combinatorial auctions with
  submodular bidders.
\newblock In {\em Proceedings of FOCS}, 2019.

\bibitem[Bal20]{Balcan-Chp20}
Maria{-}Florina Balcan.
\newblock Data-driven algorithm design.
\newblock In Tim Roughgarden, editor, {\em Beyond the Worst-Case Analysis of
  Algorithms}, pages 626--645. Cambridge University Press, 2020.

\bibitem[BDD{\etalchar{+}}21]{BalcanDDKSV21}
Maria{-}Florina Balcan, Dan~F. DeBlasio, Travis Dick, Carl Kingsford, Tuomas
  Sandholm, and Ellen Vitercik.
\newblock How much data is sufficient to learn high-performing algorithms?
  generalization guarantees for data-driven algorithm design.
\newblock In {\em Proceedings of {STOC}}, pages 919--932, 2021.

\bibitem[BHK{\etalchar{+}}24]{Banihashem-HKKO24}
Kiarash Banihashem, MohammadTaghi Hajiaghayi, Dariusz~R. Kowalski, Piotr
  Krysta, and Jan Olkowski.
\newblock Power of posted-price mechanisms for prophet inequalities.
\newblock In {\em Proceedings of SODA}, pages 4580--4604, 2024.

\bibitem[BSV18]{BalcanSV18}
Maria{-}Florina Balcan, Tuomas Sandholm, and Ellen Vitercik.
\newblock A general theory of sample complexity for multi-item profit
  maximization.
\newblock In {\em Proceedings of {EC}}, pages 173--174, 2018.

\bibitem[Buc43]{Buck43}
R.~C. Buck.
\newblock Partition of space.
\newblock {\em The American Mathematical Monthly}, 50(9):541--544, 1943.

\bibitem[CC23]{CorreaC23}
Jos{\'{e}}~R. Correa and Andr{\'{e}}s Cristi.
\newblock A constant factor prophet inequality for online combinatorial
  auctions.
\newblock In {\em Proceedings of STOC}, pages 686--697, 2023.

\bibitem[CCES20]{CorreaCES20}
Jos{\'{e}}~R. Correa, Andr{\'{e}}s Cristi, Boris Epstein, and Jos{\'{e}}~A.
  Soto.
\newblock The two-sided game of googol and sample-based prophet inequalities.
\newblock In {\em Proceedings of SODA}, pages 2066--2081, 2020.

\bibitem[CDF{\etalchar{+}}22]{CDFFLLP-SODA22}
Constantine Caramanis, Paul D{\"{u}}tting, Matthew Faw, Federico Fusco, Philip
  Lazos, Stefano Leonardi, Orestis Papadigenopoulos, Emmanouil Pountourakis,
  and Rebecca Reiffenh{\"{a}}user.
\newblock Single-sample prophet inequalities via greedy-ordered selection.
\newblock In {\em Proceedings of SODA}, 2022.

\bibitem[CDFS19]{CorreaDFS19}
Jos{\'{e}}~R. Correa, Paul D{\"{u}}tting, Felix~A. Fischer, and Kevin Schewior.
\newblock Prophet inequalities for {I.I.D.} random variables from an unknown
  distribution.
\newblock In {\em Proceedings of EC}, pages 3--17. {ACM}, 2019.

\bibitem[CHMS10]{ChawlaHMS10}
Shuchi Chawla, Jason~D. Hartline, David~L. Malec, and Balasubramanian Sivan.
\newblock Multi-parameter mechanism design and sequential posted pricing.
\newblock In {\em Proceedings of STOC}, pages 311--320, 2010.

\bibitem[Cla71]{Clarke71}
Edward~H Clarke.
\newblock Multipart pricing of public goods.
\newblock {\em Public choice}, 11(1):17--33, 1971.

\bibitem[CZ24]{CristiZilliotto24}
Andr{\'{e}}s Cristi and Bruno Ziliotto.
\newblock Prophet inequalities require only a constant number of samples.
\newblock In {\em Proceedings of STOC}, pages 491--502, 2024.

\bibitem[DFKL17]{DuettingFKL17}
Paul D\"utting, Michal Feldman, Thomas Kesselheim, and Brendan Lucier.
\newblock Prophet inequalities made easy: Stochastic optimization by pricing
  non-stochastic inputs.
\newblock In {\em Proceedings of FOCS}, 2017.

\bibitem[DK19]{DuttingK19}
Paul D{\"{u}}tting and Thomas Kesselheim.
\newblock Posted pricing and prophet inequalities with inaccurate priors.
\newblock In {\em Proceedings of EC}, pages 111--129, 2019.

\bibitem[DKL20]{DuttingKL20}
Paul D{\"{u}}tting, Thomas Kesselheim, and Brendan Lucier.
\newblock An o(log log m) prophet inequality for subadditive combinatorial
  auctions.
\newblock In {\em Proceedings of FOCS}, pages 306--317, 2020.

\bibitem[DNS05]{DobzinskiNS05}
Shahar Dobzinski, Noam Nisan, and Michael Schapira.
\newblock Approximation algorithms for combinatorial auctions with
  complement-free bidders.
\newblock In {\em Proceedings of STOC}, pages 610--618, 2005.

\bibitem[Fei09]{Feige09}
Uriel Feige.
\newblock On maximizing welfare when utility functions are subadditive.
\newblock {\em {SIAM} J. Comput.}, 39(1):122--142, 2009.

\bibitem[FGL15]{FeldmanGL15}
Michal Feldman, Nick Gravin, and Brendan Lucier.
\newblock Combinatorial auctions via posted prices.
\newblock In {\em Proceedings of SODA}, 2015.

\bibitem[FSZ16]{FeldmanSZ16}
Moran Feldman, Ola Svensson, and Rico Zenklusen.
\newblock Online contention resolution schemes.
\newblock In {\em Proceedings of SODA}, pages 1014--1033, 2016.

\bibitem[G{\etalchar{+}}73]{Groves73}
Theodore Groves et~al.
\newblock Incentives in teams.
\newblock {\em Econometrica}, 41(4):617--631, 1973.

\bibitem[GHT{\etalchar{+}}21]{GuoEtAl21}
Chenghao Guo, Zhiyi Huang, Zhihao~Gavin Tang, , and Xinzhi Zhang.
\newblock Generalizing complex hypotheses on product distributions: Auctions,
  prophet inequalities, and pandora’s problem.
\newblock In {\em Proceedings of COLT}, pages 248--2288, 2021.

\bibitem[GS20]{GS-Book20}
Anupam Gupta and Sahil Singla.
\newblock Random-order models.
\newblock In Tim Roughgarden, editor, {\em Beyond the Worst-Case Analysis of
  Algorithms}, pages 234--258. Cambridge University Press, 2020.

\bibitem[HKS07]{HajiaghayiKS07}
Mohammad~Taghi Hajiaghayi, Robert~D. Kleinberg, and Tuomas Sandholm.
\newblock Automated online mechanism design and prophet inequalities.
\newblock In {\em Proceedings of AAAI}, pages 58--65, 2007.

\bibitem[KNR20]{KaplanNR20}
Haim Kaplan, David Naori, and Danny Raz.
\newblock Competitive analysis with a sample and the secretary problem.
\newblock In {\em Proceedings of SODA}, pages 2082--2095, 2020.

\bibitem[KNR22]{KaplanNR22}
Haim Kaplan, David Naori, and Danny Raz.
\newblock Online weighted matching with a sample.
\newblock In {\em Proceedings of SODA 2022}, pages 1247--1272, 2022.

\bibitem[KRTV13]{KRTV-ESA13}
Thomas Kesselheim, Klaus Radke, Andreas T{\"{o}}nnis, and Berthold
  V{\"{o}}cking.
\newblock An optimal online algorithm for weighted bipartite matching and
  extensions to combinatorial auctions.
\newblock In {\em Proceedings of {ESA}}, volume 8125, pages 589--600, 2013.

\bibitem[KS77]{KrengelS77}
U.~Krengel and L.~Sucheston.
\newblock Semiamarts and finite values.
\newblock {\em Bulletin of the American Mathematical Society}, 83:745–747,
  1977.

\bibitem[KS78]{KrengelS78}
U.~Krengel and L.~Sucheston.
\newblock On semiamarts, amarts, and processes with finite value.
\newblock {\em Advances in Probability and Related Topics}, 4:197–266, 1978.

\bibitem[KW12]{KleinbergW12}
R.~Kleinberg and S.~M. Weinberg.
\newblock Matroid prophet inequalities.
\newblock In {\em Proceedings of STOC}, page 123–136, 2012.

\bibitem[RWW20]{RubinsteinWW20}
Aviad Rubinstein, Jack~Z. Wang, and S.~Matthew Weinberg.
\newblock Optimal single-choice prophet inequalities from samples.
\newblock In {\em Proceedings of ITCS}, volume 151 of {\em LIPIcs}, pages
  60:1--60:10, 2020.

\bibitem[SC84]{SamuelCahn84}
Esther Samuel-Cahn.
\newblock Comparison of threshold stop rules and maximum for independent
  nonnegative random variables.
\newblock {\em The Annals of Probability}, 12:1213–1216, 1984.

\bibitem[Vic61]{Vickrey61}
William Vickrey.
\newblock Counterspeculation, auctions, and competitive sealed tenders.
\newblock {\em The Journal of finance}, 16(1):8--37, 1961.

\bibitem[Wai19]{wainwright_2019}
Martin~J. Wainwright.
\newblock {\em High-Dimensional Statistics: A Non-Asymptotic Viewpoint}.
\newblock Cambridge Series in Statistical and Probabilistic Mathematics.
  Cambridge University Press, 2019.

\end{thebibliography}
\newcommand{\etalchar}[1]{$^{#1}$}

}

\end{document}